\documentclass[12pt,english]{amsart}
\usepackage{ae,aecompl}
\usepackage[T1]{fontenc}
\usepackage[latin9]{inputenc}
\usepackage[dvips,letterpaper]{geometry}
\usepackage{babel}

\newcommand{\bD}{\tilde B}
\newcommand{\Pf}{{\rm Pf}}
\newcommand{\bI}{\widetilde{\rm bott}}
\newcommand{\B}{B}
\newcommand{\Sym}{S}
\newcommand{\X}{\hat X}
\newcommand{\Y}{\hat Y}
\newcommand{\Z}{\hat Z}
\usepackage{amsthm}
\usepackage{amstext}
\usepackage{amssymb}
\usepackage[unicode=true, pdfusetitle,
 bookmarks=true,bookmarksnumbered=false,bookmarksopen=false,
 breaklinks=false,pdfborder={0 0 1},backref=false,colorlinks=false]
 {hyperref}

\numberwithin{equation}{section} 
\numberwithin{figure}{section} 
\theoremstyle{plain}
\newtheorem{thm}{Theorem}[section]
  \theoremstyle{remark}
  \newtheorem{rem}[thm]{Remark}
 \theoremstyle{definition}
  \newtheorem{example}[thm]{Example}
  \theoremstyle{plain}
  \newtheorem{lem}[thm]{Lemma}
  \theoremstyle{definition}
  \newtheorem{defn}[thm]{Definition}

\newtheorem{conjecture}{Conjecture}

\newcommand{\be}{\begin{equation}}
\newcommand{\ee}{\end{equation}}

\usepackage{graphicx}

\begin{document}

\title{Almost Commuting Matrices, Localized Wannier Functions,
and the Quantum Hall Effect}

\author{Matthew~B.~Hastings}

\author{Terry A.~Loring}

\address{Microsoft Research, Station Q, Elings Hall, University of California,
Santa Barbara, CA 93106, USA.}

\address{Department of Mathematics and Statistics, University of New Mexico,
Albuquerque, NM 87131, USA.}

\begin{abstract}
For models of non-interacting fermions moving within sites arranged
on a surface in three dimensional space, there can be obstructions to finding
localized Wannier functions.  We show that such
obstructions are
$K$-theoretic obstructions to approximating almost commuting,
complex-valued matrices by commuting matrices, and we demonstrate
numerically the presence of this obstruction for a lattice model
of the quantum Hall effect in a spherical geometry.  The numerical
calculation of the obstruction is straightforward, and does not
require translational invariance or introducing a flux torus.

We further show that there is a $Z_2$ index obstruction to
approximating almost commuting self-dual matrices by exactly
commuting self-dual matrices, and present additional conjectures regarding the
approximation of almost commuting real and self-dual matrices by
exactly commuting real and self-dual matrices.  The motivation for
considering this problem is the case of physical systems with additional
antiunitary symmetries such as time reversal or particle-hole conjugation.

Finally, in
the case of the sphere---mathematically speaking three almost commuting
Hermitians whose sum of square is near the identity---we give
the first quantitative result showing this
index is the only obstruction to finding commuting approximations.
We review the known non-quantitative results for the torus.

\end{abstract}

\maketitle

\markleft {Almost Commuting Matrices}
\markright {Almost Commuting Matrices}

\tableofcontents{}

\section{Asymptotic Commutants of Finite Rank Projections
\label{sec:AsymptoticCommutants}
}

Given a list of bounded operators on infinite dimensional Hilbert space,
it is often natural to seek a finite
rank projection $P$ that almost commutes with that set. 
The $C^{*}$-algebraist
would do so in the study of quasidiagonality, 
\cite{BrownNonQD,HadwinStronglyQD,VoiculescuQD}.
In physics, we are interested in
a projection onto a band of energy states separated from the
rest of the spectrum by an energy gap;
assuming the underlying Hamiltonian is local, this projection will itself be
local due to the gap, and hence will
approximately
commute with a list of observables. 

Whatever exact relations might
be known to hold for the original operators $(H_{1},\ldots,H_{r})$
will generally hold only approximately for the 
compressions $(PX_{1}P,\ldots,PX_{r}P).$  In a lattice model,
the projection might be from a finite dimensional space to
a space whose dimension is much lower, but still the outcome
is finite-rank operators that \emph {approximately} satisfy some relations.
Can these be approximated by finite-rank operators
that \emph {exactly} satisfy those relations.?

For example, if $X_{1}$ and $X_{2}$ in $\mathbb{B}(\mathbb{H})$
satisfy
$ -I \leq X_{j}   \leq I$ and
$ [X_{1},X_{2}]  = 0,$
then $P$ almost commuting with the $X_{j}$ implies
\begin {align*}
-I \le PX_{j}P & \leq I,\\
\left\Vert [PX_{1}P,PX_{2}P]\right\Vert  & \approx 0.
\end {align*}
We especially want to know if these almost commuting Hermitian operators
are close to commuting Hermitian operators in the corner 
$P\mathbb{B}(\mathbb{H})P\cong\mathbf {M}_{k}(\mathbb{C}).$
It is sufficient to answer this question:  can two almost commuting Hermitian
matrices be approximated by commuting Hermitian
matrices? The answer is yes. This is Lin's theorem 
\cite{LinAlmostCommutingHermitian}.

The situation very different if we consider three almost commuting
Hermitians. Specifically, Theorem~\ref {thm:cutdownExample}
shows that there are $X,Y,Z$
in $\mathbb{B}(\mathbb{H})$ satisfying
\begin {align*}
[X,Y] =[X,Z]=[Y,Z]& = 0,\\
X^2+Y^2+Z^2 & = I
\end {align*}
and finite rank projections $P_{n}$ asymptotically commuting with
$X,Y,Z$ and yet such that there do not exist  triples
$(H_{1}',H_{2}',H_{3}')$ of operators with
\begin {align*}
[H_{r}',H_{s}']  = 0,\\
(H_{1}')^{2}+(H_{2}')^{2}+(H_{3}')^{2}  = I\\
\left\Vert H_{r}-H_r'\right\Vert \approx 0,
\end {align*}
where
$H_1=PXP$, $H_2=PYP$, and $H_3=PZP$.
We will make this precise below, but the example is a variation on
the examples in
\cite{ChoiAlmostNotNearly,DavidsonAlmostCommuting,
Loring-K-thryAsymCommMatrices,VoiculescuAsymptoticallyCommuting}.

The key to showing this result is the presence of an index obstruction.
Conversely, our main quantitative result is Theorem \ref{thm:sphereQuantitative},
which gives quantitative bounds on how accurately
three almost commuting Hermitian matrices with vanishing
index obstruction can be approximated by exactly commuting
matrices.  Specifically, we show that

\begin{thm}
Suppose $(H_1,H_2,H_3)$ is a $\delta$-representation of
the sphere by matrices (defined below).  If the index
$\mathrm{bott}(H_1,H_2,H_3)$,
defined below, is vanishing, 
then there are commuting Hermitian matrices $H_1',H_2',H_3'$ with
\[ (H_1')^2 + (H_2')^2 + (H_3')^2  = I \]
and
\[ \Vert H_{r}'-H_{r}\Vert\leq\epsilon(\delta) \]
for all $r,$ where
$	\epsilon(\delta)=E(1/\delta)\delta^{1/12} $
and the function $E(x)$ grows more slowly than any power of $x.$
\end{thm}

Unless stated otherwise, matrices and vector spaces are over the complex
numbers. We use here $X^{T}$ to denote transpose and $X^{*}$ the
conjugate transpose. We always use 
$\left\Vert \mbox{ -- }\right\Vert $
to mean the operator norm. Operator means a bounded linear operator on
$\mathbb {C}^n$  or Hilbert space.  Contraction means a operator of
norm at most one.

\begin{example}
\label{exa:basicExample}
We give an example of $n$-dimensional triples of matrices which
almost commute but are not close to exactly commuting matrices.
Let
$ n=2S+1, $ where $S$ is either integer or half-integer.

Consider the spin matrices $S^1,S^2,S^3$ for a quantum spin $S$ with
$H_1=S^1/\sqrt{S(S+1)},$ 
$H_2=S^2/\sqrt{S(S+1)},$ 
$H_3=S^3/\sqrt{S(S+1)}.$  
Thus,
\[
H_1=\left[\begin{array}{ccccc}
S/\sqrt{S(S+1)}\\
 & (S-1)/\sqrt{S(S+1)}\\
 &  & \ddots\\
 &  &  & (-S+1)/\sqrt{S(S+1)}\\
 &  &  &  & -S/\sqrt{S(S+1)}\end{array}\right],
\]
and
\[
[H_r,H_s]=\epsilon^{rst}H_t/\sqrt{S(S+1)},
\]
where $\epsilon^{rst}$ is a totally anti-symmetric tensor with $\epsilon^{123}=1$.
For a mathematically oriented reader, the $S$ matrices are a representation of
the Lie algebra of $SU(2)$.

Note that
\[
H_1^2+H_2^2+H_3^2=I.
\]
It is easy to see that
\[ \Vert H_r,H_s \Vert \leq 1/S.
\]
Lemmas~\ref{lem:indexIsStable} and \ref{lem:commutingKillsIndex}
combine to tell us that
if $H_{1}',$ $H_{2}'$ and $H_{3}'$ are commuting, $n$-by-$n$
Hermitian matrices then
\[
\left\Vert H_{1}'-H_1\right\Vert 
+\left\Vert H_{2}'- H_2 \right\Vert 
+\left\Vert  H_{3}' - H_3 \right\Vert 
\geq \sqrt{1 - 4/S}.
\]
Choi \cite{ChoiAlmostNotNearly} produced a slightly better
estimate with essentially the same matrices.

These spin matrices form an example of what we call an approximate
representation of the sphere.  The mathematics oriented
reader should think about generators and relations for the
$C^*$-algebra $C_0(S^2).$  The physics oriented reader should
think that we describe the coordinates of a particle moving on the
surface of a sphere; in the presence of a magnetic field,
the particle's position is blurred out on the scale of a magnetic
length and the different coordinates cease to commute exactly.
\end{example}

To all manner of surfaces in $\mathbb {R} ^d,$ or $\mathbb {C} ^d,$
there are associated $C^*$-algebras and related collections of
almost commuting $d$-tuples of matrices.  We focus
on the surfaces most prominent in physical models:  the
disk, square, annulus, cylinder, sphere and torus.

\begin {defn}
\label {def:approxRepSphere}
Suppose $\delta \geq 0.$
A triple $(H_1, H_2, H_3)$ of operators or matrices is called a
\emph {$\delta$-representation of the sphere} if
\begin{gather*}
H_r^*   = H_r, \quad (\forall r) \\ 
\left\Vert \left[ H_r, H_s \right] \right\Vert   \leq \delta,  \quad (\forall r \neq s)\\
\left\Vert H_{1}^{2}+H_{2}^{2}+H_{3}^{2}-I\right\Vert   \leq \delta.
\end{gather*}
\end{defn}

\begin {defn}
\label {def:approxRepTorus}
Suppose $\delta \geq 0.$
A pair $(U_1, U_2)$ of operators or matrices is called a
\emph {$\delta$-representation of the torus} if
\begin{gather*}
U_r^*U_r = U_rU_r^* =I, \quad (\forall r) \\ 
\left\Vert \left[ U_1, U_2 \right] \right\Vert   \leq \delta,
\end{gather*}
\end{defn}

\begin {defn}
\label {def:approxRepSquare}
Suppose $\delta \geq 0.$
A pair $(H_1, H_2)$ of operators or matrices is called a
\emph {$\delta$-representation of the square} if
\begin{gather*}
-I \leq  H_r \leq I, \quad (\forall r) \\ 
\left\Vert \left[ H_1, H_2 \right] \right\Vert   \leq \delta, 
\end{gather*}
\end{defn}

\begin {defn}
\label {def:approxRepDisk}
Suppose $\delta \geq 0.$
An operator or matrix $X$ is called a
\emph {$\delta$-representation of the disk} if
\begin{gather*}
\left\Vert X \right\Vert   \leq 1, \\ 
\left\Vert \left[ X^*, X \right] \right\Vert   \leq \delta, 
\end{gather*}
\end{defn}

\begin {defn}
\label {def:approxRepAnnulus}
Suppose $\delta \geq 0.$
An operator or matrix $X$ is called a
\emph {$\delta$-representation of the annulus} if
\begin{gather*}
\frac{1}{2}I \leq \left| X \right|   \leq I, \\ 
\left\Vert \left[ X^*, X \right] \right\Vert   \leq \delta, 
\end{gather*}
\end{defn}

\begin {defn}
\label {def:approxRepCylinder}
Suppose $\delta \geq 0.$
A pair $(U,K)$ of operators or matrices is called a
\emph {$\delta$-representation of the cylinder} if
\begin{gather*}
U^*U = UU^* = I, \\ 
-I \leq  K \leq I,  \\ 
\left\Vert \left[ U, K \right] \right\Vert   \leq \delta, 
\end{gather*}
\end{defn}

In the above, we will usually say ``exact representation'' instead
of ``$0$-representation.''  If we have a $\delta$-representation for
$\delta > 0$ and don't wish to emphasize the exact value of $\delta,$
we will say ``approximate representation.''

We will define an invariant, called the Bott index, that applies to
approximate representations of the sphere. This is an invariant
that distinguishes those that can be approximated by exact
representations and those that cannot.

There is a more general index, defined explicitly using $K$-theory,
that applies to almost commuting triples in $C^{*}$-algebras. This
has been studied in many papers, including 
\cite{BrattElliottEvKiapproxCommUnitaries,
ElliottRordamClassificationII,LinAlmostUnitariesClassification}. 
Where possible we offer direct proofs in the language
of matrix theory, with careful error estimates and avoiding $K$-theory or
$C^{*}$-algebras.

The Bott index was discovered first in the context of
$\delta$-representations of the torus, i.e. for almost
commuting unitaries $U$ and $V.$ We give four
descriptions of this index.  Their complexity varies,
but the simplest fits in one sentence.  If $VU$ is close to
$UV$ the determinant applied to a short path between $UV$ and
$VU$ will create a closed path in the punctured plane whose
winding number equals the Bott index of $(U,V)$.

The paper is organized as follows: in the next two sections we prove results
needed for
Theorem \ref{thm:sphereQuantitative}.  In section 4 we review non-quantitative
results on the torus.  In section 5 we consider physics applications, and
in section 6 we consider index obstructions to approximation of almost commuting
real and self-dual matrices by exactly commuting real and self-dual matrices,
and we present a $Z_2$ obstruction in the self-dual case.

\section{Matrices that Almost Represent the Disk or Annulus
\label{sec:DiskAnnulus}
}

There is no second cohomology for the disk, square, annulus or 
cylinder.  This means there will be no obstruction (other than
hard work) to perturbing approximate representations to exact
representations.  We are able to make this precise in quantitative
theorems.

We start with a minor variation to the quantitative version
of Lin's theorem in \cite{hastings-2008}.

\begin{thm}
\label{thm:QuantitativeSquare} 
Suppose $(H_1,H_2)$ is a $\delta$-representation of the
square by matrices.  Then,
there exists an exact representation $(K_1,K_2)$ of the square
with
\[ \Vert H_1-K_1 \Vert, \Vert H_2-K_2 \Vert \leq \epsilon(\delta) \]
where
$ \epsilon(\delta)=E(1/\delta)\delta^{1/6} $
and the function $E(x)$ grows more slowly than any power of $x.$
\end{thm}

\begin{proof}
The only difference here is the requirement that  $K_1$ and $K_2$
be contractions, and the construction in \cite{hastings-2008}
does, in fact, produce contractions.

It is important that the function $E(x)$ does not depend on 
the dimension $n$ of the matrices.
\end{proof}

The disk is an easy to understand closed subset of the square, so
we expect an easy conversion of the quantitative almost commuting
Hermitian contractions result to a quantitative result about almost
normal contractions. This is foreshadowed by Osborne's result about
a ``bent square'' in \cite{osborne-2008}. (Osborne's result applies to
unitaries that correspond to $\delta$-representations of a subset
of the torus homeomorphic to a square.)

\begin{thm}
\label{thm:QuantitativeDisk} 
Suppose $X$ is a matrix that is a $\delta$-representation of
the disk. There exists
$X'$ that is an exact representation of the disk with
with
\[ \Vert X-X^\prime  \Vert\leq\epsilon(\delta)  \]
where
$	\epsilon(\delta)=E(1/\delta)\delta^{1/6} $
and the function $E(x)$ grows more slowly than any power of $x.$
\end{thm}

\begin{proof}
This is a corollary of Theorem~\ref{thm:QuantitativeSquare}
and any improvement on the bounds there will lead to an improvement
of the bounds here. Let $\check{\epsilon}(\delta)$ be the function
denoted $\epsilon(\delta)$ in Theorem~\ref{thm:QuantitativeDisk}.

Given a contraction $X$ with  
$\Vert[X^{*}, X]\Vert \leq \delta,$
we consider its Hermitian and anti-Hermitian parts.  These commute to
within $\delta/2$ so there are commuting Hermitian contractions
$H^{\prime}$ and $K^{\prime}$ with
\[
\left\Vert \frac{1}{2}\left(X+X^{*}\right)-H^{\prime}\right\Vert,
\left\Vert \frac{i}{2}\left(-X+X^{*}\right)-K^{\prime}\right\Vert 
\leq \check{\epsilon}(\delta/2).
\]
We set $\tilde {X}=H^{\prime}+iK^{\prime}$ and
$ X^{\prime}=f\left( \tilde {X} \right) $
for
\[
f(z)=
\begin{cases}
z & \mbox{when }|z|\leq1\\
\frac{z}{|z|} & \mbox{when }|z|>1.
\end{cases}
\]
It is clear that $X^{\prime}$ is a normal contraction. 
We easily estimate
\[
\left\Vert \tilde {X}\right\Vert  
\leq
\left\Vert X\right\Vert
+ \left\Vert \tilde {X}- X\right\Vert
\leq 1 + 2\check {\epsilon}(\delta/2).
\]
By the spectral mapping theorem
\[
\left\Vert X^{\prime}-\tilde {X}\right\Vert 
\leq
2\check {\epsilon}(\delta/2)
\]
and so 
\[
\left\Vert X^{\prime}-X\right\Vert 
\leq
4\check {\epsilon}(\delta/2).
\]
\end{proof}

\begin{thm}
\label{thm:quantitiativeAnnulus} 
Suppose $X$ is a matrix that is a $\delta$-representation of
the annulus. There exists
$X'$ that is an exact representation of the annulus with
\[ \Vert X-X^\prime  \Vert\leq\epsilon(\delta)  \]
where
$	\epsilon(\delta)=E(1/\delta)\delta^{1/6} $
and the function $E(x)$ grows more slowly than any power of $x.$
\end{thm}

\begin{proof}
Let $\check{\epsilon}(\delta)$ be the function denoted $\epsilon(\delta)$
in Theorem~\ref{thm:QuantitativeDisk}. Given $X$ with
$\Vert X \Vert \leq 1$ and 
$\left\Vert X^{-1}\right\Vert \leq 2$ 
and $\Vert[X^{*},X]\Vert \leq \delta,$
we know there is $\tilde {X}$ that is a normal contraction with
\[
\left\Vert \tilde {X} - X\right\Vert \leq \check {\epsilon}(\delta).
\]
Recall (\cite[page 177]{KadisonRingroseI})
\[
\left\Vert X-\tilde {X}\right\Vert 
< \frac{1}{2}\left\Vert X^{-1}\right\Vert ^{-1}
\implies 
\left\Vert X^{-1}-\tilde {X}^{-1}\right\Vert 
\leq
2\left\Vert X^{-1}\right\Vert ^{2}\left\Vert X-\tilde {X}\right\Vert .
\]
So long as 
$\left\Vert X-\tilde {X}\right\Vert <\frac {1}{4}$ we have
\[
\left\Vert \tilde{X}^{-1}\right\Vert  
\leq\left\Vert X^{-1}\right\Vert +\left\Vert X^{-1}-\tilde{X}^{-1}\right\Vert \\
\leq2+8\left\Vert X-\tilde{X}\right\Vert .
\]
This puts the spectrum of $\left\Vert \tilde{X}\right\Vert $ inside
an annulus with inner radius
$\left(2+8\check{\epsilon}(\delta)\right)^{-1} .$
We define $X^{\prime}=f(\tilde {X})$ for appropriate $f$ so
that
$
\frac {1}{2}\leq\left|X^{\prime}\right|\leq 1
$
and
\[
\left\Vert X^{\prime}-X\right\Vert \leq \frac{3}{2} \check {\epsilon}(\delta).
\]
\end{proof}

We can convert from the annulus to the cylinder rather easily.
The spaces are same, but the defining relations
are different.

\begin{thm}
\label{thm:cylinderQuantitative} 
Suppose $(U,K)$ is a $\delta$-representation of
the cylinder by matrices. There exists a pair
of matrices $(U^\prime, K^\prime)$ that is an exact
representation of the cylinder with 
\[
\Vert U - U^\prime  \Vert, 
\Vert K - K^\prime  \Vert \leq\epsilon(\delta) 
\]
where
$	\epsilon(\delta)=E(1/\delta)\delta^{1/6} $
and the function $E(x)$ grows more slowly than any power of $x.$
\end{thm}

\begin{proof}
Let $\check {\epsilon}(\delta)$ be the function denoted $\epsilon(\delta)$
in Theorem~\ref{thm:quantitiativeAnnulus}. 

Given a unitary $U$ and a Hermitian contraction 
$K,$ with $\Vert[U,K]\Vert\leq\delta,$
we can form
\[ X = U\left( \frac{3}{4}I + \frac{1}{4}K \right). \]
Clearly
$ \frac{1}{2}\leq\left|X\right| \leq 1 $
and
\[
\left\Vert \left[X^{*},X\right]\right\Vert  
=\left\Vert \left[\frac{3}{16}K+\frac{1}{16}K^{2},U\right]\right\Vert \\
\leq\frac{5}{16}\left\Vert \left[K,U\right]\right\Vert .
\]
We apply Theorem~\ref{thm:quantitiativeAnnulus} and obtain normal
$Y$ with 
$ \frac{1}{2}\leq\left|Y\right|\leq 1 $
and 
\[
\left\Vert Y-X\right\Vert \leq \check{\epsilon}(5\delta/16).
\]
We convert back with the polar decomposition.  Set
$ U^{\prime}=Y\left|Y\right|^{-1}  $
and
$ K^{\prime}=4\left|Y\right|-3I. $
The bounds on $\left|Y\right|$ immediately give us
$ -1\leq K^{\prime}\leq 1. $
Since $U^{\prime}$ is a unitary that commutes with $\left|Y\right|$
it commutes with $K^{\prime}.$ As to the perturbation estimates,
we see
\begin{align*}
\left\Vert K^{\prime}-K\right\Vert  
& =4\left\Vert \left|Y\right|-\left|X\right|\right\Vert \\
& =4\left\Vert U^{\prime*}Y-U^{*}X\right\Vert \\
& \leq\left\Vert U^{\prime}-U\right\Vert +\check{\epsilon}(5\delta/16).
 \end{align*}
By \cite{LiPerturbpolarDecomp} we have
\[
\left\Vert U^{\prime}-U\right\Vert 
\leq 
\frac {3}{\left\Vert X^{-1}
\right\Vert +\left\Vert Y^{-1}\right\Vert }\left\Vert X-Y\right\Vert 
\]
so
\[
\left\Vert U^{\prime}-U\right\Vert \leq 3 \check {\epsilon}(5\delta/16)
\]
and
\[
\left\Vert K^{\prime}-K\right\Vert \leq 4 \check {\epsilon}(5\delta/16).
\]
\end{proof}

\section{Matrices that Almost Represent the Sphere
\label{sec:Sphere}
}

Recall that almost commuting triples of Hermitian matrices
where the sum of squares is almost $I,$ we are regarding
as an approximate representation of the sphere.  We are still
studying a two dimensional problem, but have now a non-trivial
cohomology class to make life interesting.

\subsection{The Index, and when it Vanishes}

Useful notation here are the unitless Pauli spin matrices,
\[
\sigma_1 = \begin{bmatrix} 0 & 1\\  1 &  0  \end{bmatrix}, \quad 
\sigma_2 = \begin{bmatrix} 0 & i\\ -i &  0  \end{bmatrix}, \quad 
\sigma_3 = \begin{bmatrix} 1 & 0\\  0 & -1  \end{bmatrix}, \quad 
\]
and some indicator functions defined on the real line,
\[
f_{\gamma}(x) = \begin{cases}
		1 & \mathrm{if\ } x \geq \gamma ,\\
		0 & \mathrm{otherwise}.
	\end{cases}
\]

\begin {defn}
\label {def:BottIndex}
For any triple $(H_1,H_2,H_3)$ of Hermitians we define
first another Hermitian
\begin{align*}
\B \left(H_1, H_2, H_3 \right) 
& = \frac{1}{2}I + \frac{1}{2} \sum \sigma_r \otimes H_r \\
& = \frac{1}{2}\begin{bmatrix} 
			I + H_3 & H_1 - i H_2 \\
			H_1 + iH_2 &I - H_3  
		\end{bmatrix}.
\end {align*}
If the $H_r$  are $n$-by-$n$ matrices and $\frac{1}{2}$ 
is not in the spectrum of $\B \left(H_1, H_2, H_3 \right)$ then
the \emph{Bott index}
of this triple is the number of eigenvalues (counted according 
to multiplicity) of $ \B \left(H_1, H_2, H_3 \right) $ that are
greater than $\frac{1}{2},$ minus $n,$ so
\[
\mathrm{bott}(H_1,H_2,H_3) = \mathrm{Tr}
\left(f_{\frac{1}{2}}\left(\B \left(H_1, H_2, H_3 \right)\right)\right)
- \mathrm{Tr}\left( \begin{bmatrix} I & 0 \\ 0 & 0 \end{bmatrix}\right) 
\]
\end{defn}

The Hermitian $\B \left(H_1, H_2, H_3 \right) $ is almost idempotent.
The $K_0$ groups for $C^*$-algebras 
are generally defined in terms of projections, while for rings one uses
idempotents. For a general theory of approximate representation of
surfaces, the preferred description for indices is in terms of approximate
projections.

For this special case of the sphere, another formula demands
attention. Let
\begin{align*}
\Sym \left(H_1, H_2, H_3 \right) 
& = \sum \sigma_r \otimes H_r \\
& =		\begin{bmatrix} 
			H_3 & H_1 - i H_2 \\
			H_1 + iH_2 &- H_3  
		\end{bmatrix}.
\end{align*}
so that
\[
\mathrm{bott}(H_1,H_2,H_3) = \frac{1}{2} \mathrm{Tr}
\left(\strut f_{0}\left(\Sym \left(H_1, H_2, H_3 \right)\right)\right) .
\]

We next see that is the triple is a $\delta$-representation of the
sphere then $\delta < \frac{1}{4}$ is enough to ensure the Bott
index is defined, and as $\delta$ gets smaller the gap at
$\frac{1}{2}$ in the spectrum of the approximate projection
grows larger roughly proportionally.  For $\Sym(H_1,H_2,H_3)$
the gap is at zero.

\begin{lem}
\label{lem:boundPsquareMinusP} 
If $(H_1, H_2, H_3)$ is a $\delta$-representation of the sphere
for $\delta < \frac{1}{4}$
and $S = \Sym(H_1,H_2,H_3) $ then
$ \left\Vert \Sym^{2}-I \right\Vert \leq 4\delta  $
and
\[
\sigma(S)
\subseteq
\left[-\sqrt{1+4\delta},-\sqrt{1-4\delta} \right]
\cup
\left[\sqrt{1-4\delta},\sqrt{1+4\delta} \right]
\]
\end{lem}

\begin{proof}
From
\[
 S ^2 
= I \otimes \left(  H_1^2 + H_2^2 + H_3^2  \right)
+ \sigma_3 \otimes i \left[H_1,H_2 \right]
+ \sigma_1 \otimes i \left[H_2,H_3 \right]
+ \sigma_2 \otimes i \left[H_3,H_1 \right]
\]
we obtain the estimate
\[
\left\Vert  S ^2  - I \right\Vert
\leq 4\delta .
\]
The spectral mapping theorem tells us
\[
\sigma(S) \subseteq 
\left\{
x \in \mathbb{R}
\left| \,
|x^2 - 1| \leq 4\delta
\right.
\right\}.
\]
\end{proof}

The Bott index has appeared in many forms, under different names,
in many papers such as 
\cite{ChoiAlmostNotNearly,ExelLoringAlmostCommutingUnitary,
ExelLoringInvariats,LinAlmostUnitariesClassification,
loring1986torus,Loring-K-thryAsymCommMatrices,
LoringWhenMatricesCommute}.  See \cite{davidson2001} for
a survey of related results in operator theory.

The Bott index is clearly invariant under
conjugation by a unitary. We will see it is very stable,
is additive with respect to direct sums, and it vanishes
when the triple commutes. 

\begin{example}
\label{exa:exampleBott}

For the matrices $H_1,H_2,H_3$ of example
(\ref{exa:basicExample}), $\mathrm{bott}(H_1,H_2,H_3)=1$,
as we now show.
The approximate projector $B$ is equal to
\be
\B \left(H_1, H_2, H_3 \right) 
= \frac{1}{2}I + \frac{1}{2} 
\sum \sigma_r \otimes H_r,
\ee
which is equal to
\begin{eqnarray}
\B& = &\frac{1}{2} I + \frac{1}{2}\frac{1}{\sqrt{S(S+1)}}
\sum \sigma_r \otimes S^r.
\end{eqnarray}
This is recognizable as the Hamiltonian
describing an $SU(2)$-invariant coupling
between this spin $S$ and an
additional spin-$1/2.$  Since $\B$ commutes with total spin,
there are $2(S+1/2)+1=2S+2$ eigenvectors with spin $S+1/2$ and
$2(S-1/2)+1=2S$ eigenvectors with spin $S-1/2$.  The respective
eigenvalues are
\begin{eqnarray}
&&1/2+\frac{(S\pm 1/2)(S+1\pm 1/2)-S(S+1)-3/4}{2\sqrt{S(S+1)}}
\\ \nonumber
&=& 1/2+\frac{\pm (S+1/2)-1/2}{2 \sqrt{S(S+1)}}
\\ \nonumber
&\approx & 1/2 \pm 1/2.
\end{eqnarray}
The difference in the
number of eigenvectors with given spin is the index in this case.
\end{example}

\begin{lem}
\label{chernlemma}
If $(H_1,H_2,H_3)$ is a $\delta$-representation of the sphere
by $n$-by-$n$ matrices and $\delta<\frac{1}{4}$ then 
\[
\left \Vert
\mathrm{bott}(H_1,H_2,H_3) 
-
\frac{3}{2i}\mathrm{Tr}\left(H_1\left[H_2,H_3\right]\right)
\right \Vert
\leq 32 n\delta^2.
\]

Therefore, if $n\delta^2<1/64$, 
$\mathrm{bott}(H_1,H_2,H_3) 
=
{\rm Rnd}(\frac{3}{2i}\mathrm{Tr}\left(H_1\left[H_2,H_3\right]\right))$, 
where ${\rm Rnd}(...)$ means round to the nearest integer.
\end{lem}

\begin{proof}
Let $
p(x)=(1/2)\left(-x^{3}+3x\right).$ 
If $Q$ is a matrix with spectrum within $\gamma$ of
$\pm 1$ (for $\gamma < 1$)  then $p(Q)$ will have spectrum
within $2\gamma^2$ of $\pm1.$
Lemma~\ref{lem:boundPsquareMinusP}  tells us 
\begin{equation}
\sigma \left( \Sym \left(H_1, H_2, H_3 \right)\right)
\subseteq
\left[-1+2\delta,-1-4\delta \right]
\cup
\left[1-4\delta,1+2\delta \right]
\label{eqn:easyGapBounds}
\end{equation}
Considering the maximum possible errors on the $2n$ eigenvalues
we conclude
\[
\left \Vert
\mathrm{Tr}\left(f_0\left(\Sym \left(H_1, H_2, H_3 \right)\right)\right)
-
\mathrm{Tr}\left(p\left(\Sym \left(H_1, H_2, H_3 \right)\right)\right)
\right \Vert
\leq (2n)\left(2(4\delta)^2\right).
\]
or
\begin{equation}
\left \Vert
\mathrm{bott}(H_1,H_2,H_3)
-
\frac{1}{2}
\mathrm{Tr}\left(p\left(\Sym \left(H_1, H_2, H_3 \right)\right)\right)
\right \Vert
\leq 32 n\delta^2.
\label{eqn:thirdOrderTrace}
\end{equation}

Clearly the trace of $\Sym(H_1,H_2,H_3)$ is zero.
For the trace of the third power we can drop all terms in the
product that have two or three indices equal since the trace of
any of the Pauli spin matrices is zero, so
\begin{align*}
\mathrm{Tr}\left(S(H_{1},H_{2},H_{3})^{3}\right)
& = \sum_{r,s,t\mbox{ distinct}}\sigma_{r}
\sigma_{s}\sigma_{t}\otimes H_{r}H_{s}H_{t} \\
& = \sum_{r,s,t\mbox{ distinct}}
	\mathrm{Tr}\left(\pm iI\right)\mathrm{Tr}\left(H_rH_sH_t\right)\\
& =6i\mathrm{Tr}\left(H_1\left[H_2,H_3\right]\right).
 \end{align*}

\end{proof}

\begin{lem}
\label{lem:indexIsStable}
Suppose 
$\left(H_1, H_2, H_3 \right)$ and
$\left(K_1, K_2, K_3 \right)$ are triples of
Hermitian $n$-by-$n$ matrices.
Suppose $\left(H_1, H_2, H_3 \right)$ is a 
$\delta$-representation of the sphere with $\delta<\frac{1}{4}$.
If
\[
\left\Vert H_1-K_1\right\Vert 
+\left\Vert H_2-K_2\right\Vert
+\left\Vert H_3-K_3\right\Vert 
\leq 
\sqrt{1 - 4\delta}
\]
then the Bott index of $\left(K_1, K_2, K_3 \right)$ is defined
and
\[
\mathrm{bott}(K_1,K_2,K_3)
=
\mathrm{bott}(H_1,H_2,H_3).
\]
\end{lem}

\begin{proof}
Let $S(0) = \Sym\left(H_1, H_2, H_3 \right) $ and
$S(1) = \Sym\left(H_1, H_2, H_3 \right).$
Let 
\[
\gamma
=
\left\Vert H_{1}-K_{1}\right\Vert 
+\left\Vert H_{2}-K_{2}\right\Vert 
+\left\Vert H_{3}-K_{3}\right\Vert .
\]
Clearly
$
\left\Vert \Sym-\Sym^{\prime}\right\Vert
\leq \gamma.
$
Consider the continuous path  $S(t) = tS(1)+(1-t)S(0)$ and
notice $\Vert S(t)-S(0)\Vert \leq \gamma.$
So long as
$
\gamma  < \sqrt{1 - 4\delta}
$
the gap at zero in the spectrum at zero must persist for
all $t$ and the indices must be equal.
\end{proof}

\begin{lem}
\label{lem:indexIsAdditive}
Suppose $\delta<\frac{1}{4}$ and 
$\left(H_1, H_2, H_3 \right)$ and
$\left(K_1, K_2, K_3 \right)$ are $\delta$-representations
of the sphere by  matrices.  Then 
\begin{gather*}
\mathrm{bott}\left(
\left[\begin{array}{cc}
H_1 & 0\\
0 & K_1\end{array}\right]
,
\left[\begin{array}{cc}
H_2 & 0\\
0 & K_2\end{array}\right]
,
\left[\begin{array}{cc}
H_3 & 0\\
0 & K_3\end{array}\right]
\right)\\
=
\mathrm{bott}(H_1,H_2,H_3)
+
\mathrm{bott}(K_1,K_2,K_3).
\end{gather*}

\end{lem}

\begin{lem}
Replacing any one of the $H_{r}$ by $-H_{r}$ flips the sign of the
index.
\end{lem}

\begin{lem}
\label{lem:commutingKillsIndex}
Suppose 
$H_1,$ $H_2,$ $H_3$ are three Hermitian matrices so that
the Bott index is defined. 
If the $H_r$ pairwise commute then
\[ \mathrm{bott}(H_1,H_2,H_3)=0. \]
\end{lem}

\begin{proof}
The index is invariant under conjugation by a unitary, we may assume
the $H_r$ are diagonal. The index is additive for direct sums,
we may assume $n=1.$ For real scalars $a,$ $b$ and $c$ 
the matrix 
\[
\frac {1}{2}
\left[\begin{array}{cc}
1+a & b+ic\\
b-ic & 1-a
\end{array}\right]
\]
has eigenvalues
\[
\frac {1}{2} \pm \frac {1}{2}\sqrt {a^2+b^2+c^2}.
\]
We have one eigenvalue above $\frac{1}{2}$ so the index in
this simple case is zero.
\end{proof}

This index can vanish for another reason.  
Lemmas (\ref{lem:realsKillIndex},\ref{lem:dualsKillIndex}) consider
two cases when the index vanishes.
Later in section (\ref{sec:real}) we discusses physical
motivation for considering these cases.

\begin{lem}
\label{lem:realsKillIndex}
Suppose $\delta<\frac{1}{4}$ and 
$\left(H_1, H_2, H_3 \right)$ is a $\delta$-representation
of the sphere.  If the $H_j$ are real matrices, then
\[
\mathrm{bott}(H_1,H_2,H_3)=0.
\]

\end{lem}

This is a special case of the following.

\begin{lem}
Suppose $\delta<\frac{1}{4}$ and 
$\left(H_1, H_2, H_3 \right)$ is a $\delta$-representation
of the sphere. 
Then
\[
\mathrm{bott}(H_{1}^{T},H_{2}^{T},H_{3}^{T})=-\mathrm{bott}(H_{1},H_{2},H_{3}).
\]
\end{lem}

\begin{proof}
Consider the unitary
\[
U=\left [\begin {array}{cc}
0 & I\\
-I & 0\end {array}\right ].
\]
Then
\[
U \left (
\frac {1}{2}\left [\begin{array}{cc}
I+H_1^T & H_2^T+iH_3^T\\
H_2^{T}-iH_3^{T} & I-H_1^T
\end {array}\right]
\right)
U^{*}
+\frac {1}{2}
\left[\begin{array}{cc}
I+H_1 & H_2+iH_3\\
H_2-iH_3 & I-H_1
\end {array}\right ]^T
=
\left [\begin{array}{cc}
I & 0\\
0 & I\end{array}
\right].
\]
The number of eigenvalues near $1$ for the two summands must sum
to $2n.$ Since unitary equivalence and transpose preserve all the
necessarily-real eigenvalues, we are done.
\end{proof}

\begin{defn}
A matrix $A$ is said to be {\bf self-dual} if
$Z A^T Z =- A$,
where $Z$ is the block matrix
\be
Z=\begin{pmatrix}
0 & I \\
-I & 0
\end{pmatrix}.
\ee
\end{defn}

\begin{lem}
\label{lem:dualsKillIndex}
Suppose $\delta<\frac{1}{4}$ and 
$\left(H_1, H_2, H_3 \right)$ is a $\delta$-representation
of the sphere.  
If the $H_r$ are self-dual matrices then 
\[ \mathrm{bott}(H_1,H_2,H_3)=0. \]
\end{lem}

\begin{proof}
This is a special case of the following.
\end{proof}

\begin{lem}
Suppose $\delta<\frac{1}{4}$ and 
$\left(H_1, H_2, H_3 \right)$ is a $\delta$-representation
of the sphere by $2N$-by-$2N$  matrices.
Then 
\[
\mathrm{bott}(H_1,H_2,H_3)=-\mathrm{bott}(-Z H_1^T Z, -Z H_2^T Z, -Z H_3^T Z).
\]
\end{lem}

\begin{proof}
As proven above, ${\rm bott}(H_1,H_2,H_3)=-{\rm bott}(H_1^T,H_2^T,H_3^T)$.
So, it suffices to prove that
${\rm bott}(-Z H_1 Z, -Z H_2 Z, -Z H_3 Z)={\rm bott}(H_1,H_2,H_3)$.
However,
\begin{eqnarray}
&&\begin{pmatrix}
I- Z H_1 Z & - Z H_2 Z-i Z H_3 Z  \\
-ZH_2 Z+i Z H_3 Z & I + Z H_1 Z
\end{pmatrix}
\\ \nonumber
&=&
\begin{pmatrix} Z & 0 \\ 0 & Z \end{pmatrix}
\begin{pmatrix}
-I- H_1 & - H_2-i H_3  \\
-H_2+i H_3 & -I +H_1
\end{pmatrix}
\begin{pmatrix} Z & 0 \\ 0 & Z \end{pmatrix}
\\ \nonumber
&=&
- 
\begin{pmatrix} Z & 0 \\ 0 & Z \end{pmatrix}
\begin{pmatrix}
I+ H_1 & H_2+i H_3  \\
H_2-i H_3 & I -H_1
\end{pmatrix}
\begin{pmatrix} Z & 0 \\ 0 & Z \end{pmatrix}
\\ \nonumber
&=&
U^{*}
\begin{pmatrix}
I+ H_1 & H_2+i H_3  \\
H_2-i H_3 & I -H_1
\end{pmatrix}
U,
\end{eqnarray}
where $U$ is the unitary matrix
\be
U=i\begin{pmatrix} Z & 0 \\ 0 & Z \end{pmatrix}.
\ee
Since unitary equivalence preserves the real eigenvalues, we are done.
\end{proof}

\subsection{Cylindrical to Spherical Coordinates}

A key result, in the next subsection, is that 
when the index vanishes for matrices almost representing the sphere
they are near matrices that exactly represent the sphere.  The
proof involves a ``change of coordinates'' into spherical coordinates.

In this subsection we consider the easier change from cylindrical to
spherical.

\begin{lem}
\label{lem:toSpherical}
Suppose $\delta \geq 0.$
If  $(U,K)$ is a $\delta$-representation of the cylinder by 
$n$-by-$n$ matrices then
$\left(H_1, H_2, H_3 \right)$
is a $\delta$-representation of the sphere for
\begin{align*}
H_1 & =K,\\
H_2 & =\frac{1}{2}\left(U\sqrt{I-K}+\sqrt{I-K^{2}}U^{*}\right),\\
H_3 & =\frac{i}{2}\left(-U\sqrt{I-K^{2}}+\sqrt{I-K^{2}}U^{*}\right).
\end{align*}
\end{lem}

\begin{proof}
These matrices are evidently self-adjoint, and
\[
H_{2}+iH_{3} 
=U\sqrt{I-K^{2}}.
\]
Therefore
\[
\left(H_{2}+iH_{3}\right)^{*}\left(H_{2}+iH_{3}\right) 
=
I-K^{2}
\]
and so
\[
H_{1}^{2}+\left(H_{2}+iH_{3}\right)^{*}\left(H_{2}+iH_{3}\right)=I.
\]
Therefore
\begin{align*}
\left\Vert H_{1}^{2}+H_{2}^{2}+H_{3}^{2}-I\right\Vert  
& =\left\Vert \left[H_{2},H_{3}\right]\right\Vert .
\end{align*}
As to the commutators,
\begin{align*}
   \left\Vert \left[H_{1},H_{2}+iH_{3}\right]\right\Vert  
 & =\left\Vert KU\sqrt{I-K^{2}}-U\sqrt{I-K^{2}}K\right\Vert \\
 & \leq\left\Vert KU-UK\right\Vert 
 \end{align*}
and
\begin{align*}
\left\Vert \left[\left(H_{2}+iH_{3}\right)^{*},\left(H_{2}+iH_{3}\right)\right]\right\Vert  
& =\left\Vert \left(1-K^{2}\right)-U\left(1-K^{2}\right)U^{*}\right\Vert \\
 & =\left\Vert K^{2}U-UKK^{2}\right\Vert \\
 & \leq2\left\Vert KU-UK\right\Vert .
\end{align*}
Since
\[
\left[\left(H_{2}+iH_{3}\right)^{*},\left(H_{2}+iH_{3}\right)\right]
=
2i\left[H_{2},H_{3}\right]
\]
we see
\[
\left\Vert \left[H_{2},H_{3}\right]\right\Vert
\leq
\left\Vert \left[U,K\right]\right\Vert .
\]
 Since
\[
\left[H_{1},H_{2}+iH_{3}\right] 
 = \left[H_{1},H_{2}\right]+i\left[H_{1},H_{3}\right]
\]
we have
\[
\left\Vert \left[H_{1},H_{2}\right]+i\left[H_{1},H_{3}\right]\right\Vert 
\leq
\left\Vert \left[U,K\right]\right\Vert ,
\]
and considering real and imaginary parts, we have the weaker estimates
\begin{equation}
\left\Vert \left[H_{1},H_{2}\right]\right\Vert 
\leq
\left\Vert \left[U,K\right]\right\Vert ,
\label{eq:h1-h2}
\end{equation}
and
\begin{equation}
\left\Vert \left[H_{1},H_{3}\right]\right\Vert 
\leq
\left\Vert \left[U,K\right]\right\Vert .
\label{eq:h1-h3}
\end{equation}

\end{proof}

\subsection{Spherical to Cylindrical Coordinates}

\begin{lem}
\label{lem:toCylindrical}
Suppose $( H_1, H_2, H_3 )$
is a $\delta$-representation of the sphere by matrices.
If $\mathrm{bott}( H_1, H_2, H_3 ) = 0$
then there is a unitary $U$
so that $(U,H_3)$ is a
$\sqrt{8\delta}+2\delta$-representation of
the cylinder and
\begin{equation}
\left\Vert
U\left(I-H_1^2\right)^{\frac{1}{2}}-\left(H_1+iH_2\right)
\right\Vert 
\leq
2\sqrt{2\eta}+2\eta.
\label{eq:toCylinder-Displacement}
\end{equation}
\end{lem}

\begin{proof}
Let $P = \B(H_1, H_2, H_3).$
The index vanishing means there is a unitary $W$ so that 
\[
P=W^{*}\left[\begin{array}{cc}
D_{1} & 0\\
0 & D_{0}\end{array}\right]W
\]
for $D_{\ell}$ a diagonal matrix within $2\delta$ of $\ell.$ Therefore
\[
\left\Vert 
P  - 
W^{*}
\left[\begin{array}{cc}
I & 0\\
0 & 0
\end{array}\right]
W
\right\Vert 
\leq 2\delta.
\]

Define $A$ and $B$ by
\[
\left[\begin{array}{cc}
A & B\\
0 & 0\end{array}\right] =
\left[\begin{array}{cc}
I & 0\\
0 & 0
\end{array}\right]
W.
\]
Then
\begin{equation}
\left[\begin{array}{cc}
A & B\\
0 & 0\end{array}\right]
\left[\begin{array}{cc}
A & B\\
0 & 0
\end{array}\right]
^{*} =
\left[\begin{array}{cc}
I & 0\\
0 & 0
\end{array}\right]
\label{eqn:compareToTrivial}
\end{equation}
and 
\[
\left[\begin{array}{cc}
A & B\\
0 & 0\end{array}\right]
^{*}
\left[\begin{array}{cc}
A & B\\
0 & 0
\end{array}\right] 
=W
\left[\begin{array}{cc}
I & 0\\
0 & 0
\end{array}\right]
W^{*}
\]
so
\begin{equation}
\left\Vert 
\left[\begin{array}{cc}
A & 0\\
B & 0
\end{array}\right]
^{*}
\left[\begin{array}{cc}
A & 0\\
B & 0
\end{array}\right]
- P
\right\Vert 
\leq 2\delta.
\label{eqn:compareToBott}
\end{equation}

From (\ref{eqn:compareToTrivial}) we get one exact relation,
\[
AA^*+BB^*=I,
\]
and from (\ref{eqn:compareToBott}) several approximate relations,
\begin{equation}
\left\Vert A^{*}A-\left(\frac{1}{2}I+\frac{1}{2}H_3\right)\right\Vert 
\leq 2 \delta,
\label{eq:a_star_a}
\end{equation}
\begin{equation}
\left\Vert A^{*}B-\left(\frac{1}{2}H_1+\frac{i}{2}H_2\right)\right\Vert  
\leq 2\delta,
\label{eq:a-star-b}
\end{equation}
\begin{equation}
\left\Vert B^{*}B-\left(\frac{1}{2}I-\frac{1}{2}H_3\right)\right\Vert  
\leq 2\delta.
\label{eq:b_star-b}
\end{equation}
In particular,
\[
\left\Vert A^*A+B^*B - I \right\Vert \leq  4\delta.
\]

Recall we can insist on a unitary in the polar decomposition of a
matrix, although it may not be unique. Also notice that 
$X=U\left(X^{*}X\right)^{\frac{1}{2}}$
implies $X=\left(XX^{*}\right)^{\frac{1}{2}}U.$ See \S 83 in 
\cite{HalmosFDvectorSpaces},
for example.
Thus there are unitaries $Z$ and $V$ so that
\[
A=Z\left(A^{*}A\right)^{\frac{1}{2}}=\left(AA^{*}\right)^{\frac{1}{2}}Z
\]
and
\[
B=V\left(B^{*}B\right)^{\frac{1}{2}}=\left(BB^{*}\right)^{\frac{1}{2}}V.
\]
Next we find
\begin{align*}
&
\left\Vert
V\left(\frac{1}{2}I + \frac{1}{2}H_3\right)V^* - AA^*
\right\Vert \\
& \qquad
=\left\Vert
V\left(-\frac{1}{2}I + \frac{1}{2}H_3+B^*B\right)V^*-VB^*BV-AA^* + I
\right\Vert \\
& \qquad
\leq \left\Vert 
\frac{1}{2}I - \frac{1}{2}H_3 - B^*B\right\Vert + \left\Vert VB^*BV^*+AA^* - I
\right\Vert \\
& \qquad
=\left\Vert 
\frac{1}{2}I-\frac{1}{2}H_3 - B^*B\right\Vert + \left\Vert BB^*+AA^* - I
\right\Vert 
\end{align*}
so
\begin{equation}
\left\Vert 
V \left(\frac{1}{2}I+\frac{1}{2}H_3\right) V^* -AA^*
\right\Vert 
\leq 2\delta,
\label{eq:a-a-star}
\end{equation}
and
\[
\left\Vert 
Z^*AA^*Z - \left( \frac{1}{2}I + \frac{1}{2}H_3 \right)
\right\Vert 
=
\left\Vert
A^*A - \left( \frac{1}{2}I + \frac{1}{2}H_3 \right)
\right\Vert 
\]
so
\begin{equation}
\left\Vert
Z^*AA^*Z - \left( \frac{1}{2}I + \frac{1}{2}H_3 \right)
\right\Vert 
\leq 2\delta.
\label{eq:z-star-a-star-a-z}
\end{equation}
These two equations now tell us
\begin{align*}
& \left\Vert (Z^*V)^*H_3(Z^*V)-H_3 \right\Vert \\
& \qquad = 
\left\Vert Z \left( I+H_3 \right)Z^* - V\left(I+H_3\right)v^* \right\Vert \\
& \qquad \leq \left\Vert Z\left(I+H_{1}\right)Z^{*}-2AA^{*}\right\Vert +\left\Vert 2AA^{*}-v\left(I+H_{1}\right)V^{*}\right\Vert \\
& \qquad \leq 
8\delta.
\end{align*}
Let $U$ be the unitary $U=Z^{*}V.$ We just showed 
$(U, H_3)$ is an $8\delta$-representation of the cylinder.
From equation (\ref{eq:a-a-star}) we get
\[
\left\Vert
V \left( \frac{1}{2}I+\frac{1}{2}H_3 \right)^{\frac{1}{2}} V^*
- \left(AA^*\right)^{\frac{1}{2}}
\right\Vert 
\leq\sqrt{2\delta},
\]
c.f.\ \cite{PedersenCommutatorInequality} or 
\cite{BhatiaRajendraKittanehFuad}.
To equation (\ref{eq:b_star-b}) we apply equation (1) in 
\cite{BhatiaRajendraKittanehFuad}
to produce the estimate
\[
\left\Vert 
\left(B^*B\right)^{\frac{1}{2}} -
\left(\frac{1}{2}I-\frac{1}{2}H_3\right)^{\frac{1}{2}}
\right\Vert 
\leq \sqrt{2\eta}
\]
and now
\begin{align*}
& \frac{1}{2}
\left\Vert 
U\left(I-H_3^2\right)^{\frac{1}{2}}-\left(H_1+iH_2\right)
\right\Vert \\
& \quad =
\left\Vert 
U\left(\frac{1}{2}I+\frac{1}{2}H_3\right)^{\frac{1}{2}}\left(\frac{1}{2}I -
\frac{1}{2}H_3\right)^{\frac{1}{2}}-\left(\frac{1}{2}H_1+\frac{i}{2}H_2\right)
\right\Vert \\
& \quad \leq
\left\Vert
U
\left(\frac{1}{2}I+\frac{1}{2}H_3\right)^{\frac{1}{2}}
\left(\frac{1}{2}I-\frac{1}{2}H_3\right)^{\frac{1}{2}}
-
U
\left(\frac{1}{2}I+\frac{1}{2}H_3\right)^{\frac{1}{2}}
\left(B^*B\right)^{\frac{1}{2}}
\right\Vert \\
& \qquad
+\left\Vert Z^*V\left(\frac{1}{2}I+\frac{1}{2}H_3\right)^{\frac{1}{2}}
 \left(B^*B\right)^{\frac{1}{2}}- Z^*\left(AA^*\right)^{\frac{1}{2}}
 V\left(B^*B\right)^{\frac{1}{2}}\right\Vert \\
& \qquad
+\left\Vert A^*B-\left(\frac{1}{2}H_1+\frac{i}{2}H_2\right)\right\Vert \\
& \quad \le 
\left\Vert \left(\frac{1}{2}I-I\frac{1}{2}H_3\right)^{\frac{1}{2}} 
-\left(B^*B\right)^{\frac{1}{2}}\right\Vert \\
& \qquad
+\left\Vert V\left(\frac{1}{2}I+\frac{1}{2}H_3\right)^{\frac{1}{2}}
-\left(AA^*\right)^{\frac{1}{2}}V\right\Vert \\
& \qquad
+\left\Vert A^*B-\left(\frac{1}{2}H_1+\frac{i}{2}H_2\right)\right\Vert \\
& \leq \sqrt{8\delta} + 2\delta.
\end{align*}
\end{proof}

\begin{thm}
\label{thm:sphereQuantitative}
Suppose $(H_1,H_2,H_3)$ is a $\delta$-representation of
the sphere by matrices.  If
\[
\mathrm{bott}(H_1,H_2,H_3)=0,
\]
then there are commuting Hermitian matrices $H_1',H_2',H_3'$ with
\[ (H_1')^2 + (H_2')^2 + (H_3')^2  = I \]
and
\[ \Vert H_{r}'-H_{r}\Vert\leq\epsilon(\delta) \]
for all $r,$ where
$	\epsilon(\delta)=E(1/\delta)\delta^{1/12} $
and the function $E(x)$ grows more slowly than any power of $x.$
\end{thm}

\begin{proof}
Let $\check{\epsilon}(\delta)$ be the function denoted $\epsilon(\delta)$
from Theorem~\ref{thm:cylinderQuantitative}. 

By Lemma~\ref{lem:toCylindrical} there is a unitary $U$ so that
\[
\left\Vert U^{*}H_3U-H_3\right\Vert \leq6\delta
\]
and 
\[
\left\Vert 
U\left(I-H_3^2\right)^{\frac{1}{2}}-\left(H_1+iH_2\right)
\right\Vert
\leq
2\sqrt{6\delta}+6\delta.
\]
Theorem~\ref{thm:cylinderQuantitative} produces unitary $V$ and
Hermitian contraction $K$ that commute and with
\[
\Vert V-U\Vert\leq\check{\epsilon}\left(6\delta\right)
\]
and 
\[
\Vert K-H_3\Vert\leq\check{\epsilon}\left(6\delta\right).
\]

Define $H_1',H_2',H_3'$ by
\begin{align*}
H_1' & 
=\frac{1}{2}\left(V\sqrt{1-K^2}+\sqrt{1-K^2}V^{*}\right)
=\mathrm{Re}\left(V\sqrt{1-K^2}\right),\\
H_2' & =\frac{i}{2}\left(-V\sqrt{1-K^2}+\sqrt{1-K^2}V^{*}\right)
=\mathrm{Im}\left(V\sqrt{1-K^2}\right).,\\
H_3' & 
=K.
\end{align*}
Lemma~\ref{lem:toSpherical} implies that $V$ and $K$ commute and
satisfy
\[
\left(H_{1}^{\prime}\right)^{2}+\left(H_{2}^{\prime}\right)^{2}+\left(H_{3}^{\prime}\right)^{2} 
=I.
\]
Of course
\[
\Vert H_{1}'-H_{1}\Vert\leq\check{\epsilon}\left(6\delta\right).
\]
Finally
\begin{align*}
\left\Vert H_1+iH_2-H_1'-iH_2'\right\Vert  
& \qquad =
\left\Vert H_1+iH_2-V\sqrt{1-K^{2}}\right\Vert \\
& \qquad \leq
\left\Vert H_1+iH_2-U\sqrt{1-H_3^{2}}\right\Vert +\left\Vert U\sqrt{1-H_3^2}-V\sqrt{1-K^2}\right\Vert \\
& \qquad \leq
2\sqrt{6\delta}+6\delta+\left\Vert U-V\right\Vert +\left\Vert \sqrt{1-H_{1}^{2}}-\sqrt{1-K^2}\right\Vert \\
& \qquad \qquad \leq
2\sqrt{6\delta}+6\delta+\check{\epsilon}\left(6\delta\right)+\sqrt{\left\Vert \left(1-H_3^{2}\right)-\left(1-K^{2}\right)\right\Vert }\\
& \qquad \leq
2\sqrt{6\delta}+6\delta+\check{\epsilon}\left(6\delta\right)+\sqrt{2\left\Vert H_3-K\right\Vert }\\
& \qquad \leq
2\sqrt{6\delta}+6\delta+\check{\epsilon}\left(6\delta\right)+\sqrt{2\check{\epsilon}\left(6\delta\right)}.
\end{align*}
\end{proof}

\begin{rem}
The power of $1/12$ in Theorem~\ref{thm:sphereQuantitative} can be improved upon.
It is possible to modify the construction in \cite{hastings-2008}
to improve the approximation of one of the Hermitians as the cost
of weakening the approximation of the other. This leads to an asymmetric
version of Theorem~\ref{thm:cylinderQuantitative} and an improvement
to Theorem~\ref{thm:sphereQuantitative}.
\end{rem}

\begin{thm}
\label{thm:cutdownExample}
On Hilbert space $\mathbb{H}$ there are bounded Hermitian operators
$H_{1},$ $H_{2}$ and $H_{3}$ and finite rank projections $P_{1}\leq P_{2}\leq\ldots$
so that:
\begin{enumerate}
\item the strong limit of the $P_{n}$ is the identity $I;$
\item the $H_{r}$ commute;
\item for $r=1,2,3$ we have 
 ${\displaystyle \lim_{n\rightarrow\infty}
 \left\Vert \left[P_{n},H_{r}\right]\right\Vert }=0;$
\item if $K_{n,1},$ $K_{n,2}$ and $K_{n,3}$ are commuting Hermitian operators
then
\[
\left\Vert K_{n,1}-P_{n}H_{1}P_{n}\right\Vert +\left\Vert K_{n,2}-P_{n}H_{2}P_{n}\right\Vert +\left\Vert K_{n,3}-P_{n}H_{3}P_{n}\right\Vert \rightarrow1.
\]
\end{enumerate}
\end{thm}

\begin{proof}
The idea is to put the matrices 
$H_{n,1},$ $H_{n,2},$ $H_{n,3},$ 
from Example~\ref{exa:basicExample} 
down the diagonal of an infinite matrix, but 
``doubling'' as in \cite{hastings2008topology}.
Our almost commuting projections will cut one of the double
blocks in half.

The index of 
\[
\left[\begin{array}{cc}
H_{n,1} & 0\\
0 & -H_{n,1}
\end{array}\right],
\left[\begin{array}{cc}
H_{n,2} & 0\\
0 & -H_{n,2}
\end{array}\right],
\left[\begin{array}{cc}
H_{n,3} & 0\\
0 & -H_{n,3}
\end{array}\right]
\]
is zero, so we can approximate these by $X_n,$ $Y_n$ and $Z_n$
that are commuting $2n$-by-$2n$ Hermitians whose squares sum to
one and with
\[
\left\Vert 
X_n-
\left[\begin{array}{cc}
H_{n,1} & 0\\
0 & -H_{n,1}
\end{array}\right]
\right\Vert 
\rightarrow 0,
\]
\[
\left\Vert
Y_n-
\left[\begin{array}{cc}
H_{n,2} & 0\\
0 & -H_{n,2}
\end{array}\right]
\right\Vert 
\rightarrow 0,
\]
\[
\left\Vert Y_n-
\left[\begin{array}{cc}
H_{n,3} & 0\\
0 & -H_{n,3}
\end{array}\right]
\right\Vert \rightarrow 0.
\]
Let 
\[
Q_n = \begin {bmatrix} 
I & 0 \\
0 & 0
\end {bmatrix}
.
\]
Let $X$, $Y$, $Z$  and $P$ correspond to the block diagonal
matrices formed out of
\[
X_1, X_2,\ldots, X_{n-1}, X_n, X_{n+1}, X_{n+2},\ldots
\]
\[
Y_1, Y_2,\ldots, Y_{n-1}, Y_n, Y _{n+1}, Y _{n+2},\ldots
\]
\[
Z_1, Z_2,\ldots, Z_{n-1}, Z_n, Z_{n+1}, Z_{n+2},\ldots
\]
and
\[
I_2, I_4,  \ldots, I_{2n-2}, Q_n , 0 , 0, \ldots .
\]
We have
\[
\left\Vert XP_n-P_n X\right\Vert =
\left\Vert X_n
\left[\begin{array}{cc}
I & 0\\
0 & 0
\end{array}\right]-
\left[\begin{array}{cc}
I & 0\\
0 & 0
\end{array}\right]X_n\right\Vert \rightarrow 0
\]
and similarly $\left\Vert YP_n-P_nY\right\Vert \rightarrow 0$
and $\left\Vert ZP_n-P_nZ\right\Vert \rightarrow0.$ However,
the index of
\[
P_nXP_n,P_nYP_n,P_nZP_n
\]
is $1,$ so cannot be approximated by commuting Hermitians.
\end{proof}

\section{Matrices that Almost Represent the Torus
\label{sec:Torus}
}

\begin{defn}
Suppose $U$ and $V$ are unitaries and
$ \Vert UV - VU \Vert < 2. $
The \emph {winding number invariant} $\omega(U,V)$ of $(U,V)$ is
the winding number of the closed path in $\mathbb{C}\setminus\{0\}$
given by the formula
\[
t \mapsto \det\left((1-t)UV-tVU\right).
\]
Let $\log(\mbox{--})$ denote the branch of the logarithm continuous
except on the negative reals. Let $\mathrm{Tr}$ denote the trace
on $\mathbf{M}_{n},$ normalized by $\mathrm{Tr}(I)=n.$ The following
is due to Exel (\cite[p.\ 213]{ExelSoftTorusI}).
\end{defn}

\begin{lem}
\label{lem:traceLogIsWinding}
So long as $\left\Vert \left[U,V\right]\right\Vert <2,$
\[
\omega(U,V)=
\frac{1}{2\pi}\mathrm{Tr}\left(-i\log\left(VUV^{*}U^{*}\right)\right).
\]
\end{lem}

\begin{proof}
There is a homotopy between two paths from $UV$ to $VU.$ The first
path is the linear path. The second path sends $t$ to 
$ e^{t\log\left(VUV^{*}U^{*}\right)}UV. $
See \cite[p.\ 213]{ExelSoftTorusI} for details.
\end{proof}

The following should be obvious.

\begin{lem}
If $U$ and $V$ are commuting $n$-by-$n$ unitaries then $\omega(U,V)=0.$
\end{lem}

This invariant is very stable.

\begin{lem}
If $U_0,U_1$ and $V_0,V_1$ are $n$-by-$n$ unitaries with
\[
\max \left(
\left\Vert \left[U_0,V_0\right]\right\Vert 
,\left\Vert \left[U_1,V_1\right]\right\Vert 
\right)
\leq\delta
\]
then
\[
\left\Vert U_0-U_1\right\Vert +\left\Vert V_0-V_1\right\Vert 
<
2-\delta
\]
implies 
\[
\omega\left(U_0,V_0\right)=\omega\left(U_1,V_1\right).
\]
\end{lem}

\begin{proof}
Let
$\eta = \left\Vert U_0-U_1\right\Vert +\left\Vert V_0-V_1\right\Vert . $
Consider the following, defined for $(s,t)$ in the unit square,
\[
X_{s,t}
=
s_0t_0U_0V_0-s_0t_1V_0U_0+s_1t_0U_1V_1-s_1t_1V_1U_1,
\]
where $t_0=1-t,$ $t_1=t,$ $s_0=1-s$ and $s_1=s.$ We want
this to be invertible. The ``corners'' are unitaries, so if this stays
within distance $1$ of a corner it will be invertible.
Standard estimates show
\[
\left\Vert X_{s,t}-U_jV_j\right\Vert \leq 
s_0t_1\left\Vert \left[U_0,V_0\right]\right\Vert 
+s_1t_1\left\Vert \left[U_1,V_1\right]\right\Vert 
+s_{j+1}\left(\left\Vert U_1-U_0\right\Vert +\left\Vert V_1-V_0\right\Vert \right)
\]
\[
\left\Vert X_{s,t}-V_jU_j\right\Vert \leq 
s_0t_0\left\Vert \left[U_0,V_0\right]\right\Vert 
+s_1t_0\left\Vert \left[U_1,V_1\right]\right\Vert 
+s_{j+1}\left(\left\Vert U_1-U_0\right\Vert +\left\Vert V_1-V_0\right\Vert \right)
\]
where the $j+1$ is to be performed mod-$2.$
Putting in the assumptions on the norms we find
\[
\left\Vert X_{s,t}-U_jV_j\right\Vert \leq 
t_1\delta+s_{j+1}\eta
\]
\[
\left\Vert X_{s,t}-V_jU_j\right\Vert \leq 
t_0\delta+s_{j+1}\eta
\]
so there is always one corner to which the distance is less than 
$\frac{1}{2}\delta+\frac{1}{2}\eta.$
\end{proof}

\begin{example}
\label{exa:basicUnitaries}
Let $\omega_j=e^{\frac{2\pi ij}{n}}$
and 
\[
\Omega_{n} = 
\left[\begin{array}{ccccc}
\omega_{1}\\
 & \omega_{2}\\
 &  & \ddots\\
 &  &  & \omega_{n-1}\\
 &  &  &  & \omega_{n}
\end{array}\right],
\quad 
S_{n} =
\left[\begin{array}{ccccc}
0 &  &  &  & 1\\
1 & 0\\
 & 1 & \ddots\\
 &  & \ddots & 0\\
 &  &  & 1 & 0
 \end{array}\right].
\]
Then $ \Omega_nS_n\Omega_n^*S_n^* $ 
equals $ e^{\frac{2\pi i}{n}}I $ so
$ \omega(S_n,\Omega_n)=1. $
\end{example}

This is an old example 
\cite{ExelLoringAlmostCommutingUnitary,loring1986torus,
Loring-K-thryAsymCommMatrices,VoiculescuAsymptoticallyCommuting},
but we have a large lower bound on the distance to commuting unitaries.

\begin{thm}
If $U$ and $V$ are commuting $n$-by-$n$ unitaries then
\[
\left\Vert U-S_n\right\Vert +\left\Vert V-\Omega_n\right\Vert 
\geq
2\left(1-\sin\left(\frac{\pi}{n}\right)\right)
\]
\end{thm}

\begin{proof}
~
\[
\left\Vert U-S_n\right\Vert +\left\Vert V-\Omega_n\right\Vert 
<
2-\left\Vert \left[S_n,\Omega_n\right]\right\Vert 
=
2-\left|e^{\frac{2\pi i}{n}}-1\right|
=
2-2\sin\left(\frac{\pi}{n}\right).
\]

\end{proof}

This index shares a lot of properties with the 
Bott index from \S \ref{sec:Sphere}.
It is obviously invariant with respect to conjugation by a unitary.
Also
\[
\omega(V,U)=\omega(-U,V)=\omega(U,-V)=-\omega(U,V)
\]
and
\[
\omega\left(\left[\begin{array}{cc}
U_1\\
 & U_2\end{array}\right],\left[\begin{array}{cc}
V_1\\
 & V_2\end{array}\right]\right)
 =
 \omega\left(U_1,V_1\right)+\omega\left(U_2,V_2\right).
\]
In fact, this is equal to an invariant directly based on $K$-theory.

There are no ``really nice'' maps from the torus to the sphere, but
there are smooth maps that are one-to-one over most points of the
sphere. For present purposes, the smoothness is not so important. 

On map from the standard torus in $\mathbb{C}^2$ to the unit sphere
in $\mathbb{R}^3$ that is one-to-one over most points of the sphere,
and is piecewise smooth, is
\[
\Gamma\left(e^{2\pi i\theta_1},e^{2\pi i\theta_2}\right)
=
\left(\strut
f_1\left(\theta_1\right)
,g_1\left(\theta_1\right)+h_1\left(\theta_1\right)\cos\left(\theta_2\right)
,h_1\left(\theta_1\right)\sin\left(\theta_2\right)
\right),
\]
where
\[
f_1(x)=\begin{cases}
1-4x & \mbox{if }x\leq\frac{1}{2}\\
-3+4x & \mbox{if }x\leq\frac{1}{2}
\end{cases},
\]

\[
g_1(x)=\begin{cases}
2\sqrt{2x-4x^{2}} & \mbox{if }x\leq\frac{1}{2}\\
0 & \mbox{if }x\leq\frac{1}{2}
\end{cases},
\]

\[
h_{1}(x)=\begin{cases}
0 & \mbox{if }x\leq\frac{1}{2}\\
\sqrt{-8+24x-16x^{2}} & \mbox{if }x\leq\frac{1}{2}
\end{cases}.
\]
Notice $g_1h_1=0$ and 
$
f_1^2+g_1^2+h_1^2=1,
$
and from here it is easy to check that $\Gamma$ does take
values on the unit sphere.

The restriction of $\Gamma$ to the open set determined by
$ (1/2) < \theta_1 < 1 $
gives a bijection onto the sphere minus the poles.
Points in the complement are sent by $\Gamma$ to a half-equator.

This function give us a means to manufacture an approximate representation
of the sphere out of an approximate representation
of the torus. Unfortunately, the norm of a commutator
$\left\Vert \left[X,f(Y)\right]\right\Vert $
depends rather poorly on the norm of 
$\left\Vert \left[X,Y\right]\right\Vert $
when $X$ and $Y$ are normal and all we know about $f$ is that is
is piece-wise linear or smooth.

Suppose $U$ and $V$ are approximately commuting. If we write out
the approximate projective $\B$ based on the three almost commuting
Hermitians (in some order), it looks like
\[
Q(U,V) =
\left[\begin{array}{cc}
f(V) & g(V)+h(V)U\\
g(V)+U^{*}h(V) & I-f(V)
\end{array}\right]
\]
where $f,$ $g$ and $h$ are the continuous, real-valued functions
on the circle defined by
\[
f\left(e^{2\pi i\theta}\right) =
\frac{1}{2}-\frac{1}{2}f_{1}\left(\theta\right),
\]
$
g\left(e^{2\pi i\theta}\right)=\frac{1}{2}g_1\left(\theta\right),
$ and
$
h\left(e^{2\pi i\theta}\right)=\frac{1}{2}h_1\left(\theta\right). 
$
The three Hermitians that almost represent
the sphere are
\[
I+2f(V),2\mathrm{Re}\left(g(V)+iH(V)\right),2\mathrm{Im}\left(g(V)+iH(V)\right).
\]

\begin{lem}
There is a $\delta_{1}>0$ so that for all unitary matrices $U$ and
$U$ with 
$
\left\Vert \left[U,V\right]\right\Vert <\delta_1,
$
the Hermitian matrix $Q(U,V)$ does not have $\frac{1}{2}$ it its
spectrum.
\end{lem}

\begin{defn}
For all unitary matrices $U$ and $U$ with 
$
\left\Vert \left[U,V\right]\right\Vert <\delta_{1},
$
define $\kappa(U,V)$ as $-n$ plus the number of eigenvalues, counted
with multiplicity, of $Q(U,V)$ that are greater than $\frac{1}{2}.$ 
\end{defn}

Notice that $Q(U,V)$ varies continuously in $U$ and $V.$ There
is a variation on the index that lacks this feature, but it has cleaner
formulas.

In \cite{ExelLoringInvariats} it was determined that
there is a $\delta_0,$
again unspecified, so that
$\left\Vert \left[U,V\right]\right\Vert \leq\delta_0 $
implies 
\[
Q_{1}\left(U,V\right) =
\left[\begin{array}{cc}
\ell(V) & \sqrt{\ell(V)-\left(\ell(V)\right)^2}U\\
U^{*}\sqrt{\ell(V)-\left(\ell(V)\right)^2} & \ell(V)
\end{array}\right]
\]
has spectrum that does not contain $\frac{1}{2},$ where 
$
\ell(e^{2\pi ix})=x$ for $ x\in[0,1).$
We define $\kappa_1(U,V)$ as $-n$ plus the number of eigenvalues,
counted with multiplicity, of $Q_{1}(U,V)$ that are greater than
$\frac{1}{2}.$ The following is proven in \cite{ExelLoringInvariats}.
Taking into account Lemma~\ref{lem:traceLogIsWinding} we have four
methods to compute the index.

\begin{thm}
There is a $\delta_{3}>0$ so that for all unitary matrices $U$ and
$V$ with $\left\Vert \left[U,V\right]\right\Vert \leq\delta_{3},$
we have equality of the indices,
\[
\omega(U,V)=\kappa(U,V)=\kappa_{1}(U,V).
\]

\end{thm}

We have no quantitative version of the following, and no proof of
the following that does not utilize the theory of $C^*$-algebras.
We refer readers seeking a proof to 
\cite[Corollary M3]{GongLinAlmostMultiplicativeMorphisms}
and \cite[Theorem 6.15]{ELP-pushBusby}.

\begin{thm}
For every positive $\epsilon,$ there is a positive $\delta$ less
than $2$ so that $\left\Vert \left[U,V\right]\right\Vert \leq\delta$
and $\omega(U,V)=0$ for unitary matrices implies that there exists
commuting unitary matrices $U_1$ and $V_1$ with 
$\left\Vert U-U_1\right\Vert <\epsilon$
and $\left\Vert V-V_1\right\Vert <\epsilon.$ 
\end{thm}

\section{Applications}
\label{sec:applications}

\subsection{Lattice Problems and Definition of Wannier Functions}

We describe an application of the previous results to a physics problem:
whether or not there exist localized Wannier functions for
a two-dimensional insulator.  
We begin by describing the physical
system, and the appropriate mathematical description of this problem.
We then connect the existence of Wannier functions to the absence
of an obstruction to approximating almost commuting matrices.  See also
\cite{kitaev-2009}.

We consider a system of non-interacting fermions moving in a tight
binding model with the sites arranged on the surface of a sphere (this
topology is chosen to give a two-dimensional system with
no boundary) and with short-range hopping terms of bounded strength.

Since the fermions are non-interacting, the Hamiltonian of the system is
\be
{\mathcal H}=\sum_{i,j} \Psi_i^{\dagger} H_{ij} \Psi_j,
\ee
where $\Psi^\dagger_i,\Psi_i$ are fermion creation and annihilation
operators on site $i$.  The matrix $H$ is a Hermitian matrix.
Since we consider insulating systems, we assume that the
Hamiltonian has a spectral gap $\Delta E$, so that,
without loss of generality,
all eigenvalues are either less than or equal to $-\Delta E/2$ or  are
greater than or equal to $\Delta E/2$.
We let $P$ denote the projector onto the space spanned by the
eigenvectors with eigenvalues less than or equal to $-\Delta E/2$.

A basis of Wannier functions means an orthogonal set of states which
spans the range of $P$.  An example of a basis of Wannier functions
would simply be the set of eigenvectors with eigenvalues less than
or equal to $-\Delta E/2$.
A local basis of Wannier functions means a basis of Wannier functions
in which some locality requirements are imposed on the basis vectors:
for example, given a metric on the lattice, for each vector most of
the norm of the vector should be concentrated on sites near some given
site.  Usually, the eigenvectors of $H$ will {\it not} be local in
this sense.  Below, we give a few precise, but slightly different, mathematical definitions
of ``localized Wannier function''.

To describe the fact that the interactions are short-range, we need to
describe where each site is located on the
surface of the sphere.  To do this, we introduce matrices $\X,\Y,\Z$ which
are diagonal matrices.  For each site $i$, the corresponding
diagonal matrix elements
$\X_{ii},\Y_{ii},\Z_{ii}$ denote the $x,y,z$  positions of that
site on the surface of a sphere.  We impose the condition
\be
\label{sphere}
\X^2+\Y^2+\Z^2=L^2,
\ee
where $L$ is the radius of the sphere.
We define the distance between any two sites by:
\be
{\rm dist}(i,j)\equiv \sqrt{(\X_{ii}-\X_{jj})^2+(\Y_{ii}-\Y_{jj})^2+(\Z_{ii}-\Z_{jj})^2}.
\ee
Define a distance between a site $i$ and a triple of
coordinates $\vec x=(x,y,z)$ by
\be
{\rm dist}(i,\vec x))\equiv \sqrt{(\X_{ii}-x)^2+(\Y_{ii}-y)^2+(\Z_{ii}-z)^2}.
\ee

To mathematically describe the assumption of short-range interactions,
we assume that $H_{ij}=0$ for ${\rm dist}(i,j)>R$ for some range $R$.
To describe the assumption of bounded strength interactions,
we assume that $\Vert H \Vert \leq J$ for some interaction strength $J$.
The conditions are sufficient to imply a Lieb-Robinson bound
for the dynamics\cite{LiebRobinsonFiniteGroupVelocity,HastingsKomaSpectralGap,
NachtSimsLRbounds}.
These conditions imply that
\begin{eqnarray}
\label{bnd}
\Vert [\X,H] \Vert & \leq & 2 v_{LR},
\end{eqnarray}
for $v_{LR}=RJ$, and similar bounds for $\Vert [\Y,H] \Vert,\Vert [\Z,H] \Vert$.
The subscript $LR$ refers to Lieb-Robinson; this velocity that we define here can be shown to be an
upper bound on the velocity of propagation of excitations in according with the usual definition of a Lieb-Robinson velocity.

We will be interested in the case where $L>>v_{LR}/\Delta E$ below.
We now derive
a bound on $\Vert [P,\X] \Vert,\Vert [P,\Y] \Vert, \Vert [P,\Z] \Vert$.
Because of the spectral gap, we can write
\be
\label{gfn}
P=\Delta E \int {\rm d}t f(\Delta E t) \exp(i H t),
\ee
for any function $f(t)$ such that the Fourier
transform, $\tilde f(\omega)$ obeys
$\tilde f(\omega)=1$ for $\omega\leq -1/2$ 
and $\tilde f(\omega)=0$ for $\omega\geq 1/2$.  
Then,
\begin{eqnarray}
\label{int}
\Vert [\X,P] \Vert  &= & 
\Delta E \Vert [\X,\int {\rm d}t f(\Delta E t) \exp(i H t)] \Vert
\\ \nonumber
& \leq & \Delta E \int {\rm d}t |f(\Delta E t)| \, \Vert [\X,\exp(i H t)] \Vert
\\ \nonumber
& \leq &
2 \Delta E \int {\rm d}t |f(\Delta E t)| v_{LR} |t|.
\end{eqnarray}
We now choose any specific such $\tilde f(\omega)$ which is at least twice
times differentiable for all $\omega$, so
that $f(t)$ decays faster than $1/t^3$ for large $t$.  
Then, the integral on the last line of Eq.~(\ref{int}) converges and
\be
\label{commbound}
\Vert [\X,P] \Vert  \leq {\rm const.} \times v_{LR}/\Delta E,
\ee
for some numeric constant, and the same bound holds for 
$\Vert [\Y,P] \Vert$ and $\Vert [\Z,P] \Vert$.

We now define matrices $H_r$ by
\begin{eqnarray}
\label{h123sq}
H_1 \equiv P \X P/L, \\ \nonumber
H_2 \equiv P \Y P/L, \\ \nonumber
H_3 \equiv P \Z P/L.
\end{eqnarray}
We now show that these matrices approximately represent the sphere:
\begin{lem}
\label{spherreplem}
The matrices $H_1,H_2,H_3$ defined in Eq.~(\ref{h123sq}) form a $\delta$-representation of the
sphere with
\be
\label{rep}
\delta={\rm const.} \times (v_{LR}/L \Delta E)^2.
\ee
\begin{proof}
Define 
$\X_{11}=(1-P)\X (1-P), \X_{12}=(1-P)\X P, \X_{21}=P\X (1-P)$, and
$\X_{22}=P\X P=H_1$.  Define $\Y_{11}=(1-P) \Y (1-P)$, and so on, so that
we can write
\be
\X=\begin{pmatrix}
\X_{11} & \X_{12} \\
\X_{21} & \X_{22}
\end{pmatrix} , \quad
\Y=\begin{pmatrix}
\Y_{11} & \Y_{12} \\
\Y_{21} & \Y_{22}
\end{pmatrix} . \quad
\ee
Then, $P [\X,\Y] P=\X_{21} \Y_{12}-\Y_{21}\X_{12} + [\X_{22},\Y_{22}]$.
However, since $[\X,\Y]=0$, this means that
$[\X_{22},\Y_{22}]=-\X_{21} \Y_{12}+\Y_{21}\X_{12}$, so
$\Vert [H_1,H_2] \Vert \leq \Vert \X_{21} \Vert \Vert \Y_{12} \Vert+
\Vert \Y_{21} \Vert \Vert \X_{12} \Vert$.  Note that
$\Vert \X_{21} \Vert \leq \Vert [\X,P] \Vert \leq {\rm const.}\times
v_{LR}/\Delta E$, and similarly for 
$\Vert \X_{12} \Vert, \Vert \Y_{21} \Vert \Vert \Y_{12} \Vert$.
So,
$\Vert [H_1,H_2] \Vert \leq 
{\rm const.} \times (v_{LR}/L \Delta E)^2$.

Similar bounds hold for the commutators $\Vert [H_2,H_3] \Vert,\Vert [H_3,H_1] \Vert$.

Finally,
\begin{eqnarray}
\nonumber
\Vert H_1^2+H_2^2+H_3^2-I\Vert 
&=&
\Vert P \X (1-P) \X P+ P \Y (1-P) \Y P + P \Z (1-P) \Z P \Vert \\ 
&\leq &
\Vert P \X (1-P) \X P\Vert  +\Vert P \Y (1-P) \Y P\Vert + \Vert P \Z (1-P) \Z P \Vert \\ \nonumber
& \leq & \Vert (1-P) \X P \Vert^2+
\Vert (1-P) \Y P \Vert^2+
\Vert (1-P) \Z P \Vert^2 \\ \nonumber
& \leq & 
\Vert [P,\X] \Vert^2+
\Vert [P,\Y] \Vert^2+
\Vert [P,\Z] \Vert^2 \\ \nonumber
&\leq & {\rm const.} \times (v_{LR}/L \Delta E)^2.
\end{eqnarray}
\end{proof}
\end{lem}

From now on, we work in the subspace spanned by eigenvectors 
with eigenvalues less than or equal to $-\Delta E/2$; i.e., the subspace
onto which $P$ projects.  In this subspace,
$H_r$ are a set of almost commuting Hermitian operators that 
almost square to unity.
The question now is: do there exist a set of localized Wannier functions?
We give three possible definitions of this, and relate these
definitions to the ability to approximate $H_r$ by exactly commuting matrices.
\begin{defn}
\label{d1}
A set of {\bf exponentially localized Wannier functions with
localization length} $\xi$
is a set of orthonormal vectors, $\{v^a\}$, spanning the subspace onto
which $P$ projects, such that for each vector $v^a$ the
following property holds.  Let $v^a_i$ denote the
coefficient of $v^a$ in basis
element $i$.  Let $x^a=(v^a,\X v^a)$, $y^a=(v^a,\Y v^a)$,
and $z^a=(v^a,\Z v^a)$.  
We require that, for all $a$ and all $D$,
\be
\sum_{i,{\rm dist}(i,(x^a,y^a,z^a)\geq D} |v_i^a|\leq \exp(-D/\xi).
\ee
\end{defn}

\begin{defn}
\label{d2}
A set of {\bf Wannier functions localized to length} $l_{loc}$
is a set of orthonormal vectors, $\{v^a\}$, spanning the subspace onto
which $P$ projects, such that the following property holds.
Let $x^a=(v^a,\X v^a)$, $y^a=(v^a,\Y v^a)$, and $z^a=(v^a,\Z v^a)$.
Let $\vec x^a=(x^a,y^a,z^a))$.
Let $\X'=\sum_a x^a |v^a\rangle\langle v^a|$,
let $\Y'=\sum_a x^a |v^a\rangle\langle v^a|$,
and let $\Z'=\sum_a x^a |v^a\rangle\langle v^a|$,
We require that for any $w$ in the subspace onto which $P$ projects that
\begin{eqnarray}
\label{d2e}
|(\X-\X')) w|\leq l_{loc} |w|, \\ \nonumber
|(\Y-\Y')) w|\leq l_{loc} |w|, \\ \nonumber
|(\Z-\Z')) w|\leq l_{loc} |w|, \\ \nonumber
\end{eqnarray}
\end{defn}

\begin{defn}
\label{d3}
A set of {\bf Wannier functions weakly localized to length $l_{loc}$}
is a set of orthonormal vectors, $\{v^a\}$, spanning the subspace onto
which $P$ projects, such that for each vector $v^a$ the
following property holds:
\begin{eqnarray}
\label{d3e}
(v^a, \X^2v^a)- (v^a,\X v^a)^2\leq l_{loc}^2,\\ \nonumber
(v^a, \Y^2v^a)- (v^a,\Y v^a)^2\leq l_{loc}^2, \\ \nonumber
(v^a, \Z^2v^a)- (v^a,\Z v^a)^2\leq l_{loc}^2.
\end{eqnarray}
\end{defn}

Definition (\ref{d2}) implies definition (\ref{d3});
to see this note that Eq.~(\ref{d2e})
for $w=v^a$ implies Eq.~(\ref{d3e}).  Under one assumption,
definition (\ref{d1}) implies definition (\ref{d2}). 
This
is an assumption about the number of points in the original
lattice, as in the following lemma; this assumption
expresses the
two-dimensionality of the original problem.  This next
lemma unfortunately is fairly tedious in the details, given
the simplicity of the resulting estimate.

\begin{lem}
Suppose that, for any $\vec x=(x,y,z)$ and any $l$,
the number of sites $j$ with 
$(x-\X_{ii})^2+(y-\Y_{ii})^2+(z-\Z_{ii})^2\leq l^2$
is bounded by $c_1+c_2 l^2$, for some constants $c_1,c_2$.
Then,
\be
|(\X-\X'))w|^2 \leq
{\rm const.} \times \xi^2 [c_1+c_2\xi^2]^2
\ee
and similarly for $\Y-\Y'$ and $\Z-\Z'$.
\begin{proof}
The assumption on the number of sites in the lattice
implies a bound on the number of vectors 
$v^a$ with ${\rm dist}(\vec x,\vec x^a)\leq l$ as follows.
Define
\be
P(\vec x,l) = 
\sum_{a, {\rm dist}(\vec x,\vec x^a)\leq l} |v^a\rangle\langle v^a|.
\ee
Then, for any $D$,
\begin{eqnarray}
\label{vecdenseq}
&&{\rm Tr}(P(\vec x,l)) \\ \nonumber
&=&\sum_i 
\sum_{a}^{{\rm dist}(\vec x,\vec x^a)\leq l} |v^a_i|^2 \\ \nonumber
&=& \sum_{i}^{{\rm dist}(\vec x,i)< l+D} \;
\sum_{a}^{{\rm dist}(\vec x,\vec x^a)\leq l} |v^a_i|^2+
\sum_{i}^{{\rm dist}(\vec x,i)\geq l+D} \;
\sum_{a}^{{\rm dist}(\vec x,\vec x^a)\leq l} |v^a_i|^2 \\ \nonumber
&\leq & c_1+c_2 (l+D)^2+
\sum_{i}^{{\rm dist}(\vec x,i)\geq l+D} \;
\sum_{a}^{{\rm dist}(\vec x,\vec x^a)>D} |v^a_i|^2 \\ \nonumber
&\leq & c_1+c_2(l+D)^2+
{\rm Tr}(P(\vec x,l)) \exp(-D/\xi).
\end{eqnarray}
Picking $D=\xi$, we see that ${\rm Tr}(P(\vec x,l)$,
which is the number of such vectors $v^a$
with ${\rm dist}(\vec x,\vec x^a)\leq l$, is bounded by
${\rm const.}\times (c_1+c_2 (l+\xi)^2)$.
Next, 
for $w=\sum_a A(a) v^a$, we have 
\begin{eqnarray}
\label{longeq}
& & |(\X-\X')w|^2 \\ \nonumber
&=&  
\sum_{a \neq b} \overline{A(a)} A(b) (v^a, (\X-x^a) (\X-x^b) v^b) \\ \nonumber
&=&
\sum_{a \neq b} \overline{A(a)} A(b) \sum_i 
\overline{v^a_i}v^b_i  (\X_{ii}-x^a) (\X_{ii}-x^b)
\\ \nonumber
&\leq &
\sum_{a \neq b} |A(a) A(b)| 
\sum_i
|v^a_i v^b_i| |(\X_{ii}-x^a) (\X_{ii}-x^b)|
\\ \nonumber
&\leq &
\sum_{a \neq b} |A(a) A(b)| 
\sum_i
\exp(-{\rm dist}(i,\vec x^a)/\xi)
\exp(-{\rm dist}(i,\vec x^b)/\xi)
|(\X_{ii}-x^a) (\X_{ii}-x^b)|.
\end{eqnarray}
The sum over $i$ in the last line of the above Equation
(\ref{longeq}) is equal to
\begin{eqnarray}
\label{longeq2}
&&
\sum_{i,{\rm dist}(i,\vec x^a)\leq 2 {\rm dist}(\vec x^a,\vec x^b)}
\exp(-{\rm dist}(i,\vec x^a)/\xi)
\exp(-{\rm dist}(i,\vec x^b)/\xi)
|(\X_{ii}-x^a) (\X_{ii}-x^b)|
\\ \nonumber
&+&
\sum_{i,{\rm dist}(i,\vec x^a)>2 {\rm dist}(\vec x^a,\vec x^b)}
\exp(-{\rm dist}(i,\vec x^a)/\xi)
\exp(-{\rm dist}(i,\vec x^b)/\xi)
|(\X_{ii}-x^a) (\X_{ii}-x^b)|
\\ \nonumber
&\leq &
{\rm const.}\times \Bigl( \exp(-{\rm dist}(\vec x^a,\vec x^b)/\xi) 
{\rm dist}(\vec x^a,\vec x^b)^2 [c_1+c_2({\rm dist}(\vec x^a,\vec x^b)+\xi)^2]
\\ \nonumber
&&+
\sum_{k=1}^{\infty}
\;
\sum_{i}^{2^k< \frac{{\rm dist}(i,\vec x^a)}{{\rm dist}(\vec x^a,\vec x^b)}
  \leq 2^{k+1} }
\exp(-{\rm dist}(i,\vec x^a)/\xi)
\exp(-{\rm dist}(i,\vec x^b)/\xi)
|(\X_{ii}-x^a) (\X_{ii}-x^b)| \Bigr)
\\ \nonumber
&\leq &
{\rm const.}\times \Bigl(
\exp(-{\rm dist}(\vec x^a,\vec x^b)/\xi) 
{\rm dist}(\vec x^a,\vec x^b)^2 [c_1+c_2({\rm dist}(\vec x^a,\vec x^b)+\xi)^2]
\\ \nonumber
&& \times
(1+\sum_{k=1}^{\infty}
\exp(-2^k {\rm dist}(\vec x^a,\vec x^b)/\xi)
(2\cdot 2^k)^4) \Bigr)
\\ \nonumber
&\leq &
{\rm const.} \times \exp(-{\rm dist}(\vec x^a,\vec x^b)/\xi) 
(\xi+{\rm dist}(\vec x^a,\vec x^b))^2 [c_1+c_2({\rm dist}(\vec x^a,\vec x^b)+\xi)^2]
\end{eqnarray}
Combining (\ref{longeq}) with (\ref{longeq2}) gives
\be
|(\X-\X')w|^2 \leq {\rm const.} \times
\sum_{a \neq b} |A(a) A(b)| 
\exp(-{\rm dist}(\vec x^a,\vec x^b)/\xi) 
(\xi+{\rm dist}(\vec x^a,\vec x^b))^2 [c_1+c_2({\rm dist}(\vec x^a,\vec x^b)+\xi)^2]
\ee
Combining this with (\ref{vecdenseq}), we have:
\begin{eqnarray}
\nonumber & & |(\X-\X'))w|^2 \\ \nonumber
&\leq &
{\rm const}.\times \sum_{a \neq b} |A(a) A(b)| 
\exp(-{\rm dist}(\vec x^a,\vec x^b)/\xi) 
(\xi+{\rm dist}(\vec x^a,\vec x^b))^2 [c_1+c_2({\rm dist}(\vec x^a,\vec x^b)+\xi)^2]
\\ \nonumber
&\leq &
{\rm const}. \times \sum_{a \neq b} [(|A(a)|^2+|A(b)|^2)/2]
\exp(-{\rm dist}(\vec x^a,\vec x^b)/\xi) 
(\xi+{\rm dist}(\vec x^a,\vec x^b))^2 [c_1+c_2({\rm dist}(\vec x^a,\vec x^b)+\xi)^2]
\\ \nonumber
&= & {\rm const}.\times \sum_a |A(a)|^2 \sum_{b \neq a}
\exp(-{\rm dist}(\vec x^a,\vec x^b)/\xi) 
(\xi+{\rm dist}(\vec x^a,\vec x^b))^2 [c_1+c_2({\rm dist}(\vec x^a,\vec x^b)+\xi)^2]
\\ \nonumber
&\leq &
{\rm const.} \times \xi^2 [c_1+c_2\xi^2]^2
\end{eqnarray}
\end{proof}
\end{lem}

\subsection{Obstructions to Wannier Functions and Almost Commuting Matrices}
We claim that definition (\ref{d2}) is equivalent
to the ability to approximate $H_1,H_2,H_3$ by exactly
commuting matrices $H_1',H_2',H_3'$.  Consider first
the direction of the implication that (\ref{d2}) implies
the ability to approximate by exactly commuting matrices:
simply set 
\begin{eqnarray}
H_1'=\sum_a x^a |v^a\rangle\langle v^a|, \\ \nonumber
H_2'=\sum_a y^a |v^a\rangle\langle v^a|, \\ \nonumber
H_3'=\sum_a z^a |v^a\rangle\langle v^a|.
\end{eqnarray}
Then, 
$\Vert H_1-H_1' \Vert 
=
{\rm max}_{w, |w|=1} | (H_1-H_1') w | = {\rm max}_{w,|w|=1} |(PX P -X')w|/L= {\rm max}_{w,|w|=1} |(PX-X')w|/L 
\leq{\rm max}_{w,|w|=1} |(\X-\X') w|/L
\leq l/L$.  
Similar bounds follow for $\Vert H_2-H_2' \Vert$ and 
$\Vert H_3-H_3' \Vert$.
So, Eq.~(\ref{d2e}) implies the ability to approximate by
exactly commuting matrices up to error $l/L$.
To see the converse implication, that the ability to
approximate by exactly commuting matrices implies definition (\ref{d2}),
let the vectors $v^a$ be basis vectors in a basis in
which $H_r'$ are exactly diagonal.  Then, for any $w$,
\begin{eqnarray}
|(\X-\X') w|
&=& \sqrt{|(P\X P-\X')) w|^2 + |(1-P) \X w |^2} \\ \nonumber
&=& L \sqrt{|(H_1-H_1') w|^2+|[P,\X] w|^2} \\ \nonumber
&\leq& L \sqrt{\Vert H_1-H_1' \Vert^2+4 (v_{LR}/\Delta E)^2} |w|. 
\end{eqnarray}

Combining these implications, the presence of an
index obstruction to approximating by exactly commuting matrices
implies an obstruction to definition (\ref{d2}), which implies
(under the assumption above about the number of points in the
original lattice) an obstruction to finding exponentially
localized Wannier functions.  Conversely, the ability to
approximate $H_r$ by exactly commuting matrices implies
the ability to find Wannier functions
obeying definitions (\ref{d2},\ref{d3}).  

Combining lemma (\ref{spherreplem}) with theorem (\ref{thm:sphereQuantitative}),
the absence of an index obstruction implies the ability to find Wannier
functions localized to length $L \epsilon({\rm const}\times (v_{LR}/L\Delta E)^2)=
L(v_{LR}/L \Delta E)^{1/6}
E({\rm const.}\times L \Delta E/v_{LR})$.  This length scale is asymptotically smaller than $L$.
We leave as an open problem the question of whether the absence
of an index obstruction implies the ability to find Wannier functions localized to
length $v_{LR}/\Delta E$.
We also leave as an open problem the question of whether the absence of
an index obstruction implies the ability to find exponentially
localized Wannier functions; such a result is known in the translationally invariant case\cite{momentumtorus}.
Note in this regard that under the assumptions of finite range and interaction strength and spectral gap $\Delta E$, it
is possible to prove that the matrix elements of $P$ are
exponentially decaying: $|P_{ij}| \leq \exp(-{\rm dist}(i,j)/\xi')$,
where $\xi'$ is proportional to $v_{LR}/\Delta E$.  This proof
uses standard techniques to prove locality of correlation
functions in gapped systems
\cite{hastingsLieb-Schultz-Mattis,hastingsLocalityQuantum} and
is based on using a smoother function $f(t)$ in Eq.~(\ref{gfn}).

The index obstruction  can be computed for several examples.
In an ordinary band insulator, the obstruction vanishes.
In numerical results below, we give applications to
a lattice realization of a quantum Hall system on the
surface of the sphere where the index is non-vanishing.

\subsection{Numerical Simulations}

We have performed numerical simulations to illustrate the 
usefulness of this index obstruction for studying physical
systems such as a quantum Hall effect on the sphere.  The index 
we consider has the advantage, compared to more
usual Chern number obstructions calculated on a momentum torus\cite{momentumtorus}, that it
does not require translation invariance, and it also does not 
require averaging over parameters of the Hamiltonian
on a flux torus\cite{fluxtorus}, as in the Chern
number calculation of the Hall conductance.
Compared to the noncommutative geometry approach
discussed in \cite{BellissardNCGquantumHall}, we
have a method that will adapt to a variety of surfaces
and is applicable to finite size systems.

We have considered the following system.  The choices that we 
made are deliberately somewhat arbitrary: we wanted a modest sized
system describing free particles moving on the surface of a 
sphere in the presence of a roughly uniform magnetic field exiting the sphere,
but we wanted to illustrate the robustness of this index
even in a system with no carefully chosen symmetry.
We considered a total of 560 sites on the surface of a sphere.  The sites were
distributed on 29 different latitudes, such that all 
sites in a given latitude had the same angle from the north pole (and hence
had the same $z$ coordinate).  The angles $\theta$ 
describing the latitudes were evenly spaced from $\pi/30,2\pi/30,...,29\pi/30$.
On each latitude, the number of different sites equal to 
the floor of a constant times $\cos(\theta)$, for some constant
giving $560$ total sites, with the angles $\phi$ of the 
sites evenly spaced from $0$ to $2\pi$.  The matrices $\X,\Y,\Z$ were 
chosen to equal the $x,y,z$ coordinates of each site, with
$\X_{ii}=\sin(\theta_i)\sin(\phi_i),
\Y_{ii}=\cos(\theta_i)\sin(\phi_i),
\Z_{ii}=\cos(\phi_i)$.

The Hamiltonian $H$ was chosen so that $H_{ij}=0$ if the distance
between sites $i$ and $j$, measured
as $\sqrt{(\X_{ii}-\X_{jj})^2+(\Y_{ii}-\Y_{jj})^2+(\Z_{ii}-\Z_{jj})^2}$,
was greater than a maximum range, which we
chose to be $\sqrt{0.07}\approx 0.26$.  Otherwise, the matrix
element $H_{ij}$ was chosen to equal $-J_{ij}\exp(i\omega_{ij})$, where
$J_{ij}$ was the strength of the interaction
and $\omega_{ij}$ was a phase.  We set $J_{ij}$ equal
to $-1$ plus a constant $\delta$ times a random number chosen
independently for each pair $i,j$ and uniformly between $-0.5$ and $0.5$.
The phase $\omega_{ij}$ was chosen to mimic the effect of a 
magnetic field.  We picked
\be
\omega_{ij}=n_{monopole}*(\phi_i-\phi_j)*\cos((\theta_i+\theta_j)/2),
\ee
where $n_{monopole}$ is an integer describing the 
net flux leaving the sphere.
We chose $n_{monopole}=100$ to make the net flux slightly 
smaller than the number of sites.

\begin{figure}
\label{figspect}
\includegraphics[width=3in,angle=270]{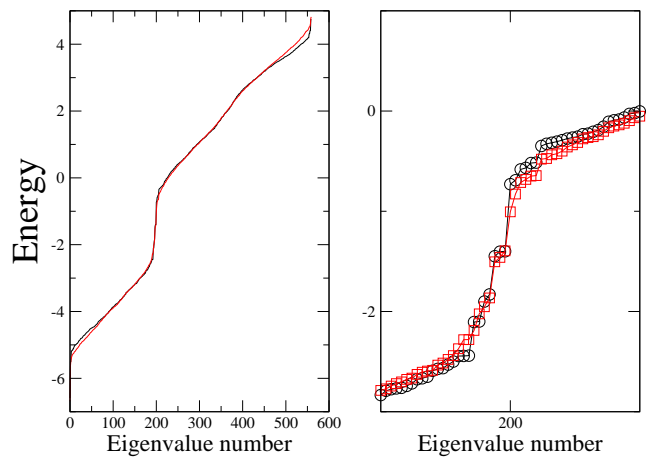}
\caption{
Spectrum of the Hamiltonian for $\delta=0$ (black line) 
and $\delta=1$ (red line).  The right column is a detail of
the figure near the band gap.
}
\end{figure}

Diagonalization of the Hamiltonian for $\delta=0,1$ gave the 
energy spectra shown in Fig.~(\ref{figspect}).
Note that for $\delta=0$ there are very few eigenvalues 
between roughly $-2.3$ and $-0.3$
and there are no eigenvalues between
$-1.3994$ and $-0.7296$.  
It is important to understand that the term ``band gap'' can be used 
in two different ways in physics.  One way means that
there is an interval $(E_{min},E_{max})$ containing no eigenvalues. 
This is referred to as a ``strict band gap.''
The other use is that there is an interval containing 
very {\it few} eigenvalues, with the corresponding eigenvectors being localized.
Such eigenvectors are referred to as ``mid-gap states.''
If there is a strict gap in the spectrum of $H$, of order unity, and if
the range of $H$ is much less than unity (for example, 
$0.07$) in our case, then the gap can be used to prove that $P$ approximately
commutes with $\X,\Y,\Z$ as discussed above.  In the example
we considered, however, there is not a very large strict band gap: if we
want to consider that the energy $1$ lies in the middle of a
strict band gap, then all we can say is that
there are are no eigenvalues in the interval
$(-1.3994,-0.7296)$.  However, even without a large strict
band gap, if the mid-gap states are indeed localized, then the
projector $P$ will still approximately commute with $\X,\Y,\Z$ and 
so the commutators of $P\X P,P\Y P,P\Z P$ will still be small.
We will see numerically below that this is the case for our problem.

There are $200$ eigenvalues less than $-1$, so the projector 
$P$ onto states with energy less than $-1$ has
rank $200$.  We computed the spectrum
of $\B$ for matrices $P\X P,P\Y P,P\Z P$.  The 
index was equal to unity, so that there were $201$ eigenvalues close to
unity and $199$ close to zero.  The smallest eigenvalue close 
to unity was equal to
$0.9763$ and the largest eigenvalue close to zero was equal to
$0.0146$, so that there is a very clear separation between 
the eigenvalues close to zero and those close to unity.
The largest commutator was $[P\Y P,P\Z P]$, with 
$\Vert [P\Y P,P\Z P] \Vert \approx 0.0298021$.  Thus, 
the matrices are very close to commuting as
claimed.  In Fig.~(\ref{figK}) we plot the eigenvalues 
of $\B$.

\begin{figure}
\label{figK}
\includegraphics[width=3in,angle=270]{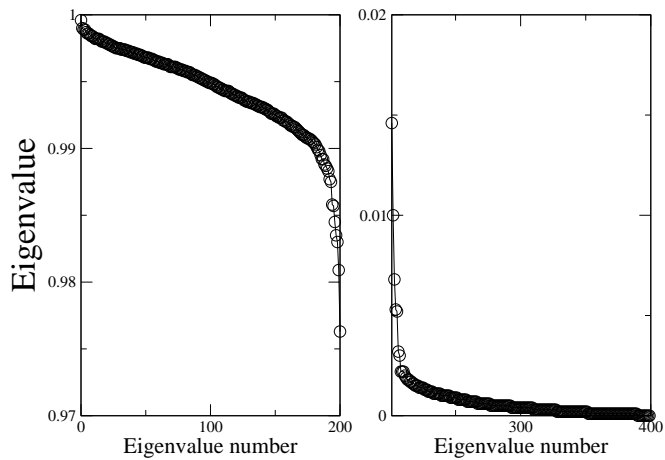}
\caption{Spectrum of $\B$ for $\delta=0$, projecting onto energies above $1$.  The left and right columns plot the eigenvalues close to unity and to zero, respectively.
There are $201$ eigenvalues on the left and $199$ on the right.}
\end{figure}

The case with $\delta=1$ was similar.  There were, in this case, 
$201$ eigenvalues less than $-1$ (so one eigenvalue crossed
unity at some value of $\delta$ between $0$ and $1$).  
However, the index remained equal to unity, and the
smallest eigenvalue of the $\B$ close to 
unity was $0.9642$ while the largest eigenvalue close to
zero was $0.0242$.
The largest commutator was again $[P\Y P,P\Z P]$, with 
$\Vert [P\Y P,P\Z P] \Vert \approx 0.036728$.  Thus, the matrices 
are very still close to commuting.

We can explicitly check that the index is insensitive to 
small changes in the energy as long as we do not 
enter the band of delocalized states.
Projecting instead onto eigenvalues less than $-2$,
there are a total of $195$ states, the index is still equal to unity, and
the eigenvalues of $\B$ which were closest to $0.5$ were
$0.305$ and $0.9562$.
Projecting onto eigenvalues less than $-3$, there are a 
total of $163$ states, the index is still equal to unity, and
the eigenvalues of $\B$ which were
closest to $0.5$ were
$0.3571$ and $0.6123$.
Projecting onto eigenvalues less than $-4$, there are 
a total of $90$ states, the index is now equal to zero, and
the eigenvalues of $\B$ which were closest to $0.5$ were
$0.4555$ and $0.6657$.  Thus, as expected, when we move 
away from the band gap, the index ceases to be well-defined since the
matrices cease to approximately commute.

Finally, we can check that we approximately represent 
the sphere, namely that $P\X P^2+P\Y P^2+P\Z P^2$ is close 
to the identity in the subspace
projected onto by $P$.  The smallest eigenvalue was 
$0.934833$ when projecting onto states with energy less
than $-1$ and the $85$-th
largest
eigenvalue was still greater than $0.99$, but the smallest eigenvalue was
only $0.447221$ when projecting onto states with energy less than $-4$.

\subsection{Relation to Hall Conductance}
The index calculated numerically for the lattice Hall system above is clearly closely related to the Hall conductance.
The formula in lemma (\ref{chernlemma}) expresses an approximation to the index in terms of a trace ${\rm Tr}(P\X P [P \Y P,P \Z P])$.
Consider a family of Hamiltonians $H$ with increasing $L$, with a uniform lower bound on the spectral gap, and uniform upper bound on $v_{LR}$.
Suppose that the dimension of the Hamiltonians, $n$, is proportional to $L^2$, as is natural for a two-dimensional system.
Then, by lemma (\ref{spherreplem}), the matrices $H_r$ form a $\delta$-representation of the sphere, with $n \delta^2\propto 1/L^2$.
This means that, by lemma (\ref{chernlemma}), ${\rm Tr}(P\X P [P \Y P,P \Z P])$ is within $1/L^2$ of an integer.
That is, this trace is approximately quantized.

This trace is closely related to the Kubo linear response formula for the Hall effect: consider applying an electric potential to
the sphere which varies uniformly between the north and south poles.  This amounts to adding a term $\Psi^\dagger \Z \Psi$ to the
Hamiltonian.  If the Hall conductance is positive, this will drive a current in the counterclockwise direction
around the sphere (and in the clockwise direction for negative Hall conductance).
Thus, at positive $x$-coordinate, one would expect to see a current in the positive $y$-direction, and at negative
$x$-coordinate one would expect to see a current in the negative $y$-direction.  If
the Hamiltonian $H$ is proportional to a projector, then this response
is proportional to the trace ${\rm Tr}(P\X P [P \Y P,P \Z P])$, up to numeric constants, and factors of the electric charge.  Thus,
we prove approximate quantization of the Hall conductance in spherical geometry for non-interacting electrons for
such Hamiltonians which are proportional to projectors.
Perhaps with more work it will be possible in this way to
prove quantization of the Hall conductance in spherical geometry for non-interacting electrons for arbitrary gapped Hamiltonians.

\section{Matrices With Additional Reality Constraints}
\label{sec:real}
In this section, we consider further the case in which the matrices
$H_r$ are assumed to be either real or self-dual.  By the results above,
since the index vanishes in this case, it is possible to approximate almost
commuting real or self-dual matrices by exactly commuting matrices.  However,
we can ask a further question: is it possible to approximate real or self-dual
matrices by exactly commuting real or self-dual matrices?

We begin with physical motivation for considering this problem.  Based
on the physical intuition, it is natural to conjecture that there is a $Z_2$
obstruction to approximating almost commuting self-dual matrices by exactly
commuting self-dual matrices, and that there are no other constructions.  We
then verify one of these conjectures: we construct this $Z_2$ index,
prove that if the index is nontrivial then there is a lower bound on the
distance to exactly commuting self-dual matrices, and construct an example
with a nontrivial index.  We then finish with precise statements of our
other conjectures.

\subsection{Physical Motivation}
In case the Hamiltonian ${\mathcal H}$ has time reversal
symmetry, the matrix $P$ will have the same symmetry. 
The possible cases of interest
correspond to different universality classes in random
matrix theories.  We discuss three classes here,
corresponding to the GUE, GOE, and GSE
classes.
For previous application of these universality classes to classifying
different insulating phases of free fermions, see
\cite{schnyder-2008}, in particular table II.

In the GUE case, ${\mathcal H}$ has no time reversal
symmetry, $P$ is a Hermitian projector, and $H_r$ are
Hermitian matrices with no further
symmetry constraints.  In the GOE case, ${\mathcal H}$ has 
time reversal symmetry, and the time reversal symmetry operator 
squares to unity: this describes
a spin-$0$ particle in a time reversal symmetric situation, 
or a spin-$1/2$ particle with time reversal symmetry and 
no spin-orbit coupling.
In this case, $H_r$ are real, symmetric matrices.  In the GSE 
case, ${\mathcal H}$ has time reversal symmetry, and the 
time reversal symmetry operator
squares to minus one.  This describes a spin-$1/2$ particle 
with strong spin-orbit coupling.  In this case, the $H_r$ 
are self-dual.

In the GSE case, the index vanishes due to the time 
reversal symmetry as explained above.
There is therefore no index obstruction to 
approximating $H_r$ by exactly commuting $H_r'$.  However,
it is natural to look for
a $Z_2$ obstruction to approximating 
three almost commuting $H_r$ by self-dual $H_r'$, because
we know that
there is a $Z_2$ index characterizing different 
translationally invariant phases of free fermions with
symplectic symmetry in two dimensions
\cite{Kane_et_al_Zmod2QuantumHall,kane_et_al_HallInGraphene}.  Such 
topologically nontrivial phases are expected, by their nature,
to be stable to perturbations of the Hamiltonian which 
break translational symmetry, as discussed in 
\cite{schnyder-2008} and have
been observed experimentally in $HgTe/(Hg,Ce)Te$ 
quantum wells\cite{konig2007quantum}.

\subsection{Index Obstruction}
We construct the index by finding a unitary transformation that
makes $S(H_1,H_2,H_3)$ anti-symmetric and then taking the sign
of the Pfaffian of this matrix.  From a physical point of view, the
existence of this unitary transformation is not surprising: the self-dual
operation can be regarded as a time-reversal symmetry operation, and a
similar time-reversal symmetry can be applied to the $\sigma$ matrices
used to construct $S(H_1,H_2,H_3)$.  Under these combined time reversal
symmetries, $S(H_1,H_2,H_3)$ changes sign; however, since there are
two spin-$1/2$s, the time reversal symmetry operator squares to unity and hence,
up to a basis change, is equivalent to transposition.  We now show this.

\begin{lem}
Let $H_r$ be self-dual.  Define the matrix
$\bD(H_1,H_2,H_3)$ by
\be
\bD(H_1,H_2,H_3)=U^* S(H_1,H_2,H_3) U,
\ee
where the unitary $U$ is defined by
\begin{eqnarray}
U=\frac{1}{\sqrt{2}} (I+Z \otimes \sigma_2).
\end{eqnarray}
Then, the matrix $\bD$ is anti-symmetric.
\begin{proof}
Note that for any $r$, $\sigma_r^T=-\sigma_2 \sigma_r \sigma_2$, while
$H_r^T=-Z H_r Z$, by the assumption of self-duality.  Also, $U^T=U$.  Thus,
\begin{eqnarray}
\bD(H_1,H_2,H_3)^T& = &U S(H_1,H_2,H_3)^T U^* \\ \nonumber
&=& U \Bigl( \sum_r H_r \otimes \sigma_r \Bigr)^T U^* \\ \nonumber
&=& U \Bigl( \sum_r H_r^T \otimes \sigma_r^T \Bigr) U^* \\ \nonumber
&=& U \Bigl( \sum_r ZH_r Z \otimes \sigma_2 \sigma_r \sigma_2 \Bigr) U^* \\ \nonumber
&=& U \Bigl( Z\otimes \sigma_2\Bigr)  \Bigl(\sum_r H_r \otimes \sigma_r \Bigr) \Bigl(Z \otimes \sigma_2 \Bigr)U^*.
\end{eqnarray}
Note that $(Z \otimes \sigma_2)^2=-I$, so
\begin{eqnarray}
U \Bigl( Z\otimes \sigma_2\Bigr) &=& 
-\frac{1}{\sqrt{2}}(I-Z \otimes \sigma_2) \\ \nonumber
&=& -U^*,
\end{eqnarray}
and
\begin{eqnarray}
\Bigl( Z\otimes \sigma_2\Bigr) U^*&=& 
\frac{1}{\sqrt{2}}(I+Z \otimes \sigma_2) U^*\\ \nonumber
&=& U.
\end{eqnarray}
Thus,
\begin{eqnarray}
U \Bigl( Z\otimes \sigma_2\Bigr)  \Bigl(\sum_r H_r \otimes \sigma_r \Bigr) \Bigl(Z \otimes \sigma_2 \Bigr)U^*
&=& - U^* \Bigl(\sum_r H_r \otimes \sigma_r \Bigr) U \\ \nonumber
&=& -\bD(H_1,H_2,H_3).
\end{eqnarray}
\end{proof}
\end{lem}

\begin{defn}
We define the index $\bI(H_1,H_2,H_3)$ for self-dual matrices $H_r$ by
\be
\bI(H_1,H_2,H_3)={\rm sgn}(
\Pf(\bD(H_1,H_2,H_3))),
\ee
where $\Pf$ is the Pfaffian and ${\rm sgn}(x)=1$ for $x>0$ and ${\rm sgn}(x)=-1$ for $x<0$.  If $\Pf(\bD(H_1,H_2,H_3))=0$, the index
$\bI(H_1,H_2,H_3)$ is not defined.
\end{defn}

\begin{lem}
\label{lem:pathstable}
Consider any continuous path of self-dual matrices, $H_r(s)$, where $s$ is a real number, $0\leq s \leq 1$.  Suppose that for all $s$,
the matrix $B(H_1,H_2,H_3)$ has non-vanishing determinant.  Then, $\bI(H_1(0),H_2(0),H_3(0))=\bI(H_1(1),H_2(1),H_3(1))$.
\begin{proof}
The determinant of $B(H_1,H_2,H_3)$ is equal to $\Pf(\bD(H_1,H_2,H_3))^2$.  As long as the determinant does not vanish, the Pfaffian
does not vanish and hence cannot change sign.
\end{proof}
\end{lem}

\begin{lem}
If $H_1,H_2,H_3$ are self-dual and exactly commuting and $\bI(H_1,H_2,H_3)$ is defined, then
$\bI(H_1,H_2,H_3)=1$.
\begin{proof}
Note that $\bI$ is multiplicative under direct sum of matrices.  Also, $\bI$ is invariant under symplectic transformation of
the $H_r$ (these transformations preserve the property of being self-dual).  Finally, for any commuting self-dual matrices $H_r$,
we can find a symplectic transformation which makes the $H_r$ diagonal; the diagonal entries of these matrices come in pairs which are
equal.  That is, if the $H_r$ are $2n$ dimensional matrices, then after this symplectic transformation then each $H_r$ is equal to
the direct sum of $n$ different $2$-by-$2$ matrices which are proportional to the identity.
Thus, it suffices to consider the case in which the $H_r$ are {\em real} scalar multiples of
the $2$-by-$2$ identity matrix. 
If $H_r = \alpha_r I$  then
\begin{eqnarray}
\bD(H_1,H_2,H_3)&=&
\frac{1}{2} (I-Z \otimes \sigma_2) \Bigl( \sum \alpha_r \otimes \sigma_r \Bigr)
(I+Z \otimes \sigma_2) \\ \nonumber
&=& i\alpha_3 Z \otimes \sigma_3
+ i\alpha_2 I \otimes \sigma_2
+ i\alpha_3 Z \otimes \sigma_1
\\ \nonumber
&=& \begin{pmatrix}
0 & i\alpha_1 & i\alpha_2 & i\alpha_3 \\
-i\alpha_1 & 0 & -i\alpha_3 & i\alpha_2 \\
i\alpha_2 & i\alpha_3 & 0 & -i\alpha_1 \\
-i\alpha_3 & i\alpha_2 & i\alpha_1 & 0
\end{pmatrix}
\end{eqnarray}
This matrix has Pfaffian equal to 
\[
(i\alpha_1)(-i\alpha_1) - (i\alpha_2)(i\alpha_2) + (i\alpha_3)(-i\alpha_3) 
= \alpha_1^2 + \alpha_2^2 + \alpha_3^2 .
\]
\end{proof}
\end{lem}

We now show that for any $\delta>0$, there exist self-dual matrices $H_r$ which form a $\delta$-representation of the
sphere with $\bI(H_1,H_2,H_3)=-1$.  This example essentially consists of two copies of
the matrices in (\ref{exa:basicExample}), with opposite Bott index between the two copies.

\begin{example}
\label{selfdualexample}
Let
$ n=2S+1. $
Consider the spin matrices $S^1,S^2,S^3$ for a quantum spin $S$ with
$H_1=I^{(1)} \otimes S^1/\sqrt{S(S+1)}$,
$H_2=\sigma_2^{(1)} \otimes S^2/\sqrt{S(S+1)}$,
$H_3=I^{(1)} \otimes S^3/\sqrt{S(S+1)}$.
We will be using the $\sigma$ matrices in two different ways in this example: first, to form the $2n$ dimensional matrices
$H_r$ from the $n$ dimensional matrices $S^r$, second to form the matrix $\bD$.  We use $\sigma^{(1)}$ to refer to the first case, and
$\sigma^{(2)}$ to refer to the second.  Note that $Z=-i\sigma_2^{(1)}\otimes I$.   We use $I^{(1)}$ to refer to the $2$-by-$2$ identity matrix.
The matrix $S^2$ is anti-symmetric, while $S^1,S^3$ are symmetric.

It is easy to see that
\[ \Vert H_r,H_s \Vert \leq 1/S,
\]
so that the matrices form a $\delta$-representation of the sphere with $\delta=1/S$.

We claim that $\bI(H_1,H_2,H_3)=-1$.
\begin{proof}
We now compute the Pfaffian.  Since the index depends only on the sign of the Pfaffian, we ignore constant factors which are real and positive.
We have
\begin{eqnarray}
&&\bD(H_1,H_2,H_3) \\ \nonumber
&=& {\rm const.} \times
(I-Z \otimes \sigma^{(2)}_2) B(H_1,H_2,H_3)
(I+Z \otimes \sigma^{(2)}_2) \\ \nonumber
&=& {\rm const.}\times \sigma^{(1)}_2 \otimes \Bigl( S^x \otimes \sigma^{(2)}_3+S^y \otimes \sigma^{(2)}_2 - S^z \otimes \sigma^{(2)}_1 \Bigr).
\end{eqnarray}
Unitarily conjugate this matrix $\bD(H_1,H_2,H_3)$ by the orthogonal transformation $(1/\sqrt{2}) I \otimes (I^{(2)}+i\sigma^{(2)}_2)$, giving
the matrix
${\rm const.} \times \sigma^{(1)}_2 \otimes \sum_r S^r \sigma^{(2)}_r$.  Since this orthogonal transformation has determinant $+1$, the Pfaffian is
unchanged.

The matrix
$\sigma^{(1)}_2 \otimes \sum_r S^r \sigma^{(2)}_r$ is anti-symmetric and equals
\be
\begin{pmatrix}
0 & i\sum_r S^r \sigma_r \\
-i \sum_r S^r \sigma_r & 0
\end{pmatrix}
\ee

The Pfaffian of this matrix is equal to the determinant of
$i\sum_r S^r \sigma_r$, which is equal to $i^{4S+2}$ times the determinant of $\sum_r S^r \sigma_r$.  The matrix
$\sum_r S^r \sigma_r$ has $2(S+1/2)+1=2S+2$ positive eigenvalues and $2(S-1/2)+1=2S$ negative eigenvalues, as
computed in example (\ref{exa:basicExample}).  Thus, the sign of the Pfaffian is equal to
\be
i^{4S+2} (-)^{2S}=-1.
\ee

\end{proof}
\end{example}

\begin{lem}
\label{lem:indexDIsStable}
Suppose 
$\left(H_1, H_2, H_3 \right)$ and
$\left(K_1, K_2, K_3 \right)$ are triples of
self-dual, Hermitian $n$-by-$n$ matrices and suppose
$\left(H_1, H_2, H_3 \right)$ is a 
$\delta$-representation
of the sphere with $\delta<1/4$.
If
\[
\left\Vert H_1-K_1\right\Vert 
+\left\Vert H_2-K_2\right\Vert
+\left\Vert H_3-K_3\right\Vert 
\leq 
\sqrt{1-4\delta}
\]
then 
\[
\bI(K_1,K_2,K_3)=
\bI(H_1,H_2,H_3)
\]
\end{lem}

\begin{proof}
We showed in the proof of Lemma \ref{lem:indexIsStable}
that the line segment from $\Sym(H_1,H_2,H_3)$ to 
$\Sym(K_1,K_2,K3)$ passes through invertibles.  It follows
that the line segment from $\bD(H_1,H_2,H_3)$ to 
$\bD(K_1,K_2,K3)$ passes through skew-symmetric invertibles
and so by lemma (\ref{lem:pathstable})
the Pfaffian does not change sign.
\end{proof}

As a corollary, the distance in operator norm 
from the matrices in Example (\ref{selfdualexample}) to the
nearest exactly commuting triple of self-dual matrices is at least $\sqrt{1-4/S}$.

\subsection{Conjectures}
We believe that the absence of this index 
obstruction implies that is possible to approximate almost commuting
self-dual matrices by exactly commuting self-dual matrices.
Before addressing this issue,
 we need to understand the two matrix 
case in the presence of additional symmetry.
We thus raise the following conjectures which 
generalize Lin's theorem:
\begin{conjecture}
For all $\epsilon>0$, there exists a $\delta>0$ such that, 
given any real, symmetric matrices $A,B$ with 
$\Vert [A,B] \Vert \leq \delta$
and $\Vert A \Vert,\Vert B \Vert \leq 1$,
there exist real, symmetric matrices $A',B'$, with 
$[A',B']=0$ and $\Vert A-A'\Vert, \Vert B-B'\Vert \leq \epsilon$.
\end{conjecture}

\begin{conjecture}
For all $\epsilon>0$, there exists a $\delta>0$ such that, 
given any self-dual, Hermitian matrices $A,B$ with 
$\Vert [A,B] \Vert \leq \delta$
and $\Vert A \Vert,\Vert B \Vert \leq 1$,
there exist self-dual, Hermitian matrices $A',B'$, 
with $[A',B']=0$ and $\Vert A-A'\Vert, \Vert B-B'\Vert \leq \epsilon$.
\end{conjecture}

The conjectures we make regarding index obstructions are:
\begin{conjecture}
For all $\epsilon>0$, there exists a $\delta>0$ such that, 
given any real, symmetric matrices $H_1,H_2,H_3$ with 
$\Vert [H_1,H_2] \Vert, \Vert [H_2,H_3] \Vert, \Vert [H_3,H_1] \Vert
 \leq \delta$, and
\be
H_1^2+H_2^2+H_3^2=I,
\ee
there exist real, symmetric commuting matrices $H_1',H_2',H_3'$, with
$\Vert H_1-H_1'\Vert, \Vert H_2-H_2'\Vert, \Vert H_3-H_3' \Vert \leq \epsilon$.
\end{conjecture}

\begin{conjecture}
For all $\epsilon>0$, there exists a $\delta>0$ such that, 
given any self-dual Hermitian matrices $H_1,H_2,H_3$ with $\bI(H_1,H_2,H_3)=1$
and
$\Vert [H_1,H_2] \Vert, \Vert [H_2,H_3] \Vert, \Vert [H_3,H_1] \Vert
 \leq \delta$,
and
\be
H_1^2+H_2^2+H_3^2=I,
\ee
there exist self-dual Hermitian commuting matrices $H_1',H_2',H_3'$, with
$\Vert H_1-H_1'\Vert, \Vert H_2-H_2'\Vert, \Vert H_3-H_3' \Vert \leq \epsilon$.
\end{conjecture}

\section{Discussion}

We have given quantitative error bounds on the ability to approximate three almost commuting matrices by three exactly commuting
matrices under the assumption of a vanishing index assumption.  We have related the ability to approximate these matrices to the
ability to find localized Wannier functions in a physical system, and we have demonstrated that it is readily possible to numerically
calculate this index for such systems.  

We have constructed a $Z_2$ index obstruction to approximation of
almost commuting self-dual matrices by exactly commuting self-dual
matrices, and
raised additional conjectures regarding almost commuting
matrices in the case of {\it real} $C^*$-algebras.
Finally, we note that table II
of
\cite{schnyder-2008} lists $10$ different classes of matrices and different
index possibilities in various dimensions.  We believe more generally that
each of the obstructions in $d=2$ in this table will correspond to a particular
obstruction to approximating almost commuting matrices; i.e., we expect
that for almost commuting matrices in class DIII there will be a $Z_2$
obstruction to approximating them by exactly commuting matrices of class
DIII.  Such obstruction is left for future work.

\bibliographystyle{plain}
\bibliography{cstarRefs}

\begin{thebibliography}{10}

\bibitem{fluxtorus}
J.~E. Avron and R.~Seiler.
\newblock Quantization of the hall conductance for general multiparticle
  schrodinger operators.
\newblock {\em Phs. Rev. Lett.}, 54:259, 1985.

\bibitem{BellissardNCGquantumHall}
J.~Bellissard, A.~van Elst, and H.~Schulz-Baldes.
\newblock The noncommutative geometry of the quantum {H}all effect.
\newblock {\em J. Math. Phys.}, 35(10):5373--5451, 1994.

\bibitem{BhatiaRajendraKittanehFuad}
Rajendra Bhatia and Fuad Kittaneh.
\newblock Some inequalities for norms of commutators.
\newblock {\em SIAM J. Matrix Anal. Appl.}, 18(1):258--263, 1997.

\bibitem{BrattElliottEvKiapproxCommUnitaries}
Ola Bratteli, George~A. Elliott, David~E. Evans, and Akitaka Kishimoto.
\newblock Homotopy of a pair of approximately commuting unitaries in a simple
  {$C\sp *$}-algebra.
\newblock {\em J. Funct. Anal.}, 160(2):466--523, 1998.

\bibitem{momentumtorus}
C.~Brouder, G.~Panati, M.~Calandra, C.~Mourougane, and Marzari N.
\newblock Exponential localization of wannier functions in insulators.
\newblock {\em Phys. Rev. Lett.}, 98:046402, 2007.

\bibitem{BrownNonQD}
Nathanial~P. Brown.
\newblock On quasidiagonal {$C\sp *$}-algebras.
\newblock In {\em Operator algebras and applications}, volume~38 of {\em Adv.
  Stud. Pure Math.}, pages 19--64. Math. Soc. Japan, Tokyo, 2004.

\bibitem{ChoiAlmostNotNearly}
Man~Duen Choi.
\newblock Almost commuting matrices need not be nearly commuting.
\newblock {\em Proc. Amer. Math. Soc.}, 102(3):529--533, 1988.

\bibitem{DavidsonAlmostCommuting}
Kenneth~R. Davidson.
\newblock Almost commuting {H}ermitian matrices.
\newblock {\em Math. Scand.}, 56(2):222--240, 1985.

\bibitem{davidson2001}
Kenneth~R. Davidson and Stanislaw~J. Szarek.
\newblock Local operator theory, random matrices and {B}anach spaces.
\newblock In {\em Handbook of the geometry of {B}anach spaces, {V}ol. {I}},
  pages 317--366. North-Holland, Amsterdam, 2001.

\bibitem{ELP-pushBusby}
S{\o}ren Eilers, Terry~A. Loring, and Gert~K. Pedersen.
\newblock Morphisms of extensions of {$C\sp *$}-algebras: pushing forward the
  {B}usby invariant.
\newblock {\em Adv. Math.}, 147(1):74--109, 1999.

\bibitem{ElliottRordamClassificationII}
George~A. Elliott and Mikael R{\o}rdam.
\newblock Classification of certain infinite simple {$C\sp *$}-algebras. {II}.
\newblock {\em Comment. Math. Helv.}, 70(4):615--638, 1995.

\bibitem{ExelSoftTorusI}
Ruy Exel.
\newblock The soft torus and applications to almost commuting matrices.
\newblock {\em Pacific J. Math.}, 160(2):207--217, 1993.

\bibitem{ExelLoringAlmostCommutingUnitary}
Ruy Exel and Terry Loring.
\newblock Almost commuting unitary matrices.
\newblock {\em Proc. Amer. Math. Soc.}, 106(4):913--915, 1989.

\bibitem{ExelLoringInvariats}
Ruy Exel and Terry~A. Loring.
\newblock Invariants of almost commuting unitaries.
\newblock {\em J. Funct. Anal.}, 95(2):364--376, 1991.

\bibitem{GongLinAlmostMultiplicativeMorphisms}
Guihua Gong and Huaxin Lin.
\newblock Almost multiplicative morphisms and almost commuting matrices.
\newblock {\em J. Operator Theory}, 40(2):217--275, 1998.

\bibitem{HadwinStronglyQD}
Don Hadwin.
\newblock Strongly quasidiagonal {$C\sp *$}-algebras.
\newblock {\em J. Operator Theory}, 18(1):3--18, 1987.
\newblock With an appendix by Jonathan Rosenberg.

\bibitem{HalmosFDvectorSpaces}
Paul~R. Halmos.
\newblock {\em Finite-dimensional vector spaces}.
\newblock Springer-Verlag, New York, second edition, 1974.
\newblock Undergraduate Texts in Mathematics.

\bibitem{hastings-2008}
M.~B. Hastings.
\newblock Making almost commuting matrices commute.
\newblock {\em Communications in Mathematical Physics}, 291(2):321--345, 2009.

\bibitem{HastingsKomaSpectralGap}
Matthew~B. Hastings and Tohru Koma.
\newblock Spectral gap and exponential decay of correlations.
\newblock {\em Comm. Math. Phys.}, 265(3):781--804, 2006.

\bibitem{hastingsLieb-Schultz-Mattis}
MB~Hastings.
\newblock {Lieb-Schultz-Mattis in higher dimensions}.
\newblock {\em Physical Review B}, 69(10):104431, 2004.

\bibitem{hastingsLocalityQuantum}
MB~Hastings.
\newblock {Locality in quantum and Markov dynamics on lattices and networks}.
\newblock {\em Physical review letters}, 93(14):140402, 2004.

\bibitem{hastings2008topology}
MB~Hastings et~al.
\newblock {Topology and phases in fermionic systems}.
\newblock {\em J. Stat. Mech}, page L01001, 2008.

\bibitem{KadisonRingroseI}
Richard~V. Kadison and John~R. Ringrose.
\newblock {\em Fundamentals of the theory of operator algebras. {V}ol. {I}},
  volume 100 of {\em Pure and Applied Mathematics}.
\newblock Academic Press Inc. [Harcourt Brace Jovanovich Publishers], New York,
  1983.
\newblock Elementary theory.

\bibitem{kane_et_al_HallInGraphene}
CL~Kane and EJ~Mele.
\newblock {Quantum spin Hall effect in graphene}.
\newblock {\em Physical review letters}, 95(22):226801, 2005.

\bibitem{Kane_et_al_Zmod2QuantumHall}
CL~Kane and EJ~Mele.
\newblock {Z\_ $\{$2$\}$ Topological Order and the Quantum Spin Hall Effect}.
\newblock {\em Physical review letters}, 95(14):146802, 2005.

\bibitem{kitaev-2009}
Alexei Kitaev.
\newblock Periodic table for topological insulators and superconductors.
\newblock http://arxiv.org/abs/0901.2686, to appear in the Proceedings of the
  L.D.Landau Memorial Conference "Advances in Theoretical Physics", June 22-26,
  2008, Chernogolovka, Moscow region, Russia, 2009.

\bibitem{konig2007quantum}
M.~Konig, S.~Wiedmann, C.~Brune, A.~Roth, H.~Buhmann, L.W. Molenkamp, X.L. Qi,
  and S.C. Zhang.
\newblock {Quantum spin Hall insulator state in HgTe quantum wells}.
\newblock {\em Science}, 318(5851):766, 2007.

\bibitem{LiPerturbpolarDecomp}
W.~Li.
\newblock {On the perturbation bound in unitarily invariant norms for
  subunitary polar factors}.
\newblock {\em Linear Algebra and Its Applications}, 429(2-3):649--657, 2008.

\bibitem{LiebRobinsonFiniteGroupVelocity}
Elliott~H. Lieb and Derek~W. Robinson.
\newblock The finite group velocity of quantum spin systems.
\newblock {\em Comm. Math. Phys.}, 28:251--257, 1972.

\bibitem{LinAlmostCommutingHermitian}
Huaxin Lin.
\newblock Almost commuting selfadjoint matrices and applications.
\newblock In {\em Operator algebras and their applications ({W}aterloo, {ON},
  1994/1995)}, volume~13 of {\em Fields Inst. Commun.}, pages 193--233. Amer.
  Math. Soc., Providence, RI, 1997.

\bibitem{LinAlmostUnitariesClassification}
Huaxin Lin.
\newblock Almost commuting unitaries and classification of purely infinite
  simple {$C\sp *$}-algebras.
\newblock {\em J. Funct. Anal.}, 155(1):1--24, 1998.

\bibitem{loring1986torus}
Terry~A. Loring.
\newblock {\em {The torus and noncommutative topology}}.
\newblock PhD thesis, University of California, Berkeley, 1986.

\bibitem{Loring-K-thryAsymCommMatrices}
Terry~A. Loring.
\newblock {$K$}-theory and asymptotically commuting matrices.
\newblock {\em Canad. J. Math.}, 40(1):197--216, 1988.

\bibitem{LoringWhenMatricesCommute}
Terry~A. Loring.
\newblock When matrices commute.
\newblock {\em Math. Scand.}, 82(2):305--319, 1998.

\bibitem{NachtSimsLRbounds}
Bruno Nachtergaele and Robert Sims.
\newblock Lieb-{R}obinson bounds and the exponential clustering theorem.
\newblock {\em Comm. Math. Phys.}, 265(1):119--130, 2006.

\bibitem{osborne-2008}
T.J. Osborne.
\newblock {Almost commuting unitaries with spectral gap are near commuting
  unitaries}.
\newblock {\em Proc. Amer. Math. Soc.}, 137:4043--4048, 2009.

\bibitem{PedersenCommutatorInequality}
Gert~K. Pedersen.
\newblock A commutator inequality.
\newblock In {\em Operator algebras, mathematical physics, and low-dimensional
  topology ({I}stanbul, 1991)}, volume~5 of {\em Res. Notes Math.}, pages
  233--235. A K Peters, Wellesley, MA, 1993.

\bibitem{schnyder-2008}
A.P. Schnyder, S.~Ryu, A.~Furusaki, and A.W.W. Ludwig.
\newblock {Classification of Topological Insulators and Superconductors}.
\newblock In {\em AIP Conference Proceedings}, volume 1134, page~10, 2009.

\bibitem{VoiculescuAsymptoticallyCommuting}
Dan Voiculescu.
\newblock Asymptotically commuting finite rank unitary operators without
  commuting approximants.
\newblock {\em Acta Sci. Math. (Szeged)}, 45(1-4):429--431, 1983.

\bibitem{VoiculescuQD}
Dan Voiculescu.
\newblock Around quasidiagonal operators.
\newblock {\em Integral Equations Operator Theory}, 17(1):137--149, 1993.

\end{thebibliography}

\end{document}